\documentclass[11pt]{article}
\usepackage{geometry}
\usepackage{multirow}
\makeatletter
\newif\ifCLASSOPTIONcompsoc \CLASSOPTIONcompsocfalse
\newif\ifCLASSINFOpdf \CLASSINFOpdftrue
\newcommand{\IEEEdisplaynontitleabstractindextext}{}
\newcommand{\IEEEpeerreviewmaketitle}{}
\newcommand{\IEEEraisesectionheading}[1]{#1}
\newcommand{\appendices}{\appendix}
\newcommand{\IEEEtitleabstractindextext}[1]{}
\makeatother

\usepackage{url}
\usepackage{amsthm}
\usepackage{hyperref}
\usepackage{xcolor}
\usepackage{algorithm}
\usepackage{siunitx}
\usepackage[noend]{algpseudocode}
\usepackage[font=footnotesize,labelfont=bf,labelsep=period]{caption}
\usepackage{amsfonts}
\ifCLASSOPTIONcompsoc
  \usepackage[nocompress]{cite}
\else
  \usepackage{cite}
\fi
\ifCLASSINFOpdf
\else
\fi
\usepackage{tikz}
\usepackage{threeparttable}
\theoremstyle{plain}
\newtheorem{theorem}{Theorem}
\newtheorem{lemma}{Lemma}

\theoremstyle{definition}

\theoremstyle{remark}

\usepackage{amsmath}
\usepackage{float}
\usepackage{booktabs}
\usepackage[utf8]{inputenc}

\usepackage{array}

\usepackage{graphicx}

\DeclareMathOperator*{\argmin}{\mathrm{arg\,min}}

\begin{document}
%
\title{Adaptive Benign Overfitting (ABO): \\
Overparameterized RLS for Online Learning in Non-stationary Time-series
}
%
%
%
%

\author{
Luis Ontaneda Mijares \\
University College London \\
\texttt{luis.ont.mij@gmail.com}
\and
Nick Firoozye \\
Department of Computer Science, University College London \\
\texttt{n.firoozye@ucl.ac.uk}
}

\date{}

\maketitle

\begin{abstract}
Overparameterized models have recently challenged conventional learning theory by exhibiting improved generalization beyond the interpolation limit, a phenomenon known as \emph{benign overfitting}.  
This work introduces \emph{Adaptive Benign Overfitting (ABO)}, extending the recursive least-squares (RLS) framework to this regime through a numerically stable formulation based on orthogonal--triangular updates.  
A QR-based exponentially weighted RLS (QR–EWRLS) algorithm is introduced, combining random Fourier feature mappings with forgetting-factor regularization to enable online adaptation under non-stationary conditions.  
The orthogonal decomposition prevents the numerical divergence associated with covariance-form RLS while retaining adaptability to evolving data distributions.  
Experiments on nonlinear synthetic time series confirm that the proposed approach maintains bounded residuals and stable condition numbers while reproducing the double-descent behavior characteristic of overparameterized models. Applications to forecasting foreign exchange and electricity demand show that ABO is highly accurate (comparable to baseline kernel methods) while achieving speed improvements of between 20\% to 40\%.  The results provide a unified view linking adaptive filtering, kernel approximation, and benign overfitting within a stable online learning framework.
\end{abstract}

\newpage


\IEEEdisplaynontitleabstractindextext

%
\IEEEpeerreviewmaketitle

\IEEEraisesectionheading{\section{Introduction}\label{sec:introduction}}

Adaptive filtering, or equivalently online regression, underpins modern signal processing, control, and sequential prediction.
As modern systems  operate increasingly in high-dimensional and non-stationary environments, classical stability–bias trade-offs no longer retain their relevance; it is possible to interpolate without overfitting.   In this work, we revisit the Recursive Least Squares (RLS) framework through the lens of overparameterization, linking adaptive filtering with the recent theory of benign overfitting, thereby deriving a fast and practical way of doing online inference in large feature models for non-stationary time series.

\subsection{Benign Overfitting}

The phenomenon of \emph{benign overfitting}—where models generalize well despite perfectly interpolating their training data—has reshaped our understanding of generalization in high-dimensional learning.  
The associated \emph{double-descent} behavior, in which test error decreases again beyond the interpolation threshold, was first observed in deep neural networks \cite{zhang2016understanding,belkin2019reconciling,bartlett2020benign} and later formalized for linear and kernel models \cite{hastie2022surprises,mei2022generalization}.  
Even in simple least-squares regression, overparameterization can yield low risk (test-set error
as measured by the MSE) under mild spectral conditions on the covariance matrix.

In linear regression, benign overfitting is often analyzed through the minimum-norm, or ridgeless, estimator
\begin{equation}
  \hat{\beta}
   = \lim_{\lambda \to 0} \argmin_{\beta} \bigl(
    \sum_t \|y_t - X_t \beta\|^2 + \lambda \|\beta\|^2\bigr)
\end{equation}
which is equivalent to
\begin{equation}
  \hat{\beta} = \argmin_{\beta} \|\beta\|^2
  \quad \text{s.t. } y = X\beta
\end{equation}
In the case where \(X\) is a \(N \times D\) Gaussian i.i.d.\ random matrix with \(N/D \to \gamma < 1\) as \(N, D \to \infty\), the test error can decrease as \(D\) increases, provided that the spectrum of \(X^\top X\) remains well-behaved\footnote{with a sizeable number of large singular values and an arbitrarily large number of small, but not vanishing 
singular values}.  
This phenomenon reveals that overparameterization can act as an implicit regularizer, replacing explicit shrinkage with spectral smoothing.

Although the analysis of benign overfitting was first established in static settings, analogous behavior arises naturally in adaptive and online learning.  
Recursive Least Squares (RLS) and its exponentially weighted variant (EWRLS) \cite{haykin2002adaptive,sayed2003fundamentals} continuously update regression coefficients as new data arrive, effectively performing a streaming analogue of ridge regression.  
Their stability and their ability to generalize when overparameterized remain key questions motivating this work.

These findings, in the context of non-stationary time series, may explain why increasing model width, reducing effective condition numbers, improves generalization, particularly when noise directions are broad enough to absorb variance orthogonal to the principal components.

\subsection{From Kernel Regression to Neural Tangent Kernels}

Before introducing neural and random-feature analogues, it is helpful to recall that least-squares estimation can be viewed as kernel regression with a linear kernel.  
Extending this concept to nonlinear feature maps yields kernel methods,
where generalization is governed by the eigenvalue spectrum of the kernel matrix
rather than the raw design matrix~\cite{principe2011kernel,hastie2009elements}.
This kernel perspective forms the foundation for both the \emph{Neural Tangent Kernel (NTK)} and its efficient approximation via \emph{Random Fourier Features (RFFs)}.

\subsection{From Neural Tangent Kernels to Random Features}

The Neural Tangent Kernel (NTK) framework \cite{jacot2018ntk} interprets the gradient-descent dynamics of the training of DNNs 
in the infinite-width limit as kernel regression with a specific data-dependent kernel in a Reproducing Kernel Hilbert Space (RKHS).
The convergence and generalization properties of overparameterized neural networks can thus be analyzed through the spectral properties of the corresponding NTKs.
This viewpoint provides a direct bridge between deep learning and classical kernel theory, explaining why overparameterized neural networks behave like data-dependent kernel machines.

Kernel methods, however, scale poorly with dataset size due to the need to compute and invert large kernel matrices and 
a computationally tractable analogue was introduced by \emph{Random Fourier Features} (RFFs) \cite{rahimi2007random}, which approximate shift-invariant kernels through random sinusoidal bases derived from B\"{o}chner’s theorem.  
Given features \( z(x)\in\mathbb{R}^D \), regression proceeds via
\[
  \min_{\beta}\|y - Z\beta\|^2 + \lambda \|\beta\|^2
\]
where \(Z\) has fixed width independent of sample size.  
RFFs therefore retain much of the expressive power of kernel methods while remaining computationally efficient. The study of RFFs in the overparameterized regime has revealed parallels with NTKs and deep networks, including benign overfitting and double-descent phenomena.

\subsection{Double Descent in RFFs}

Random Fourier Features can be interpreted as two-layer neural networks
with trigonometric activation functions; alternative nonlinearities such
as ReLU and leaky-ReLU are also viable~\cite{sun2018reluRFF}. Their
computational simplicity, together with universal-approximation
guarantees for infinitely wide networks~\cite{cybenko1989approximation,
hornik1991approximation}, makes RFFs an ideal setting for studying
behavior in the interpolating and non-classical regimes. In particular,
RFFs exhibit the same \emph{double-descent} phenomenon observed in
overparameterized neural networks~\cite{hastie2022surprises,
bartlett2020benign}.

To set notation, the minimum-norm interpolator is
\[
\hat{y} = X\hat{\beta},
\qquad
\hat{\beta} = (X^{\top}X)^{\dagger} X^{\top}y
\]

A convenient unifying viewpoint comes from ridge regression:
\begin{align*}
    \hat{\beta}
        &= \lim_{\lambda \to 0^{+}} (X^{\top}X + \lambda I)^{-1} X^{\top} y \\
        &= X^{\dagger} y 
\end{align*}
This limit always returns the Moore--Penrose minimum-norm solution and
therefore covers both regimes:
\begin{itemize}
    \item \textbf{Classical regime ($D \le N$):}
          \(X^\top X\) is nonsingular and
          \(X^\dagger = (X^\top X)^{-1} X^\top\)
    \item \textbf{Overparameterized regime ($D > N$):}
          \(X^\top X\) is singular and the pseudoinverse is required.
\end{itemize}
In both cases the fundamental identity
\[
    (X^\top X)^{\dagger} X^\top = X^{\dagger}
\]
ensures a consistent form of the estimator.

Empirically, RFFs display the characteristic double-descent curve:
\begin{itemize}
  \item training error decreases monotonically until the interpolation
        threshold;
  \item test error follows the classical U-shape, peaking at
        interpolation;
  \item beyond interpolation, test error decreases again as the model
        becomes increasingly overparameterized.
\end{itemize}

\begin{figure}[t]
  \centering
  \includegraphics[scale=0.4]{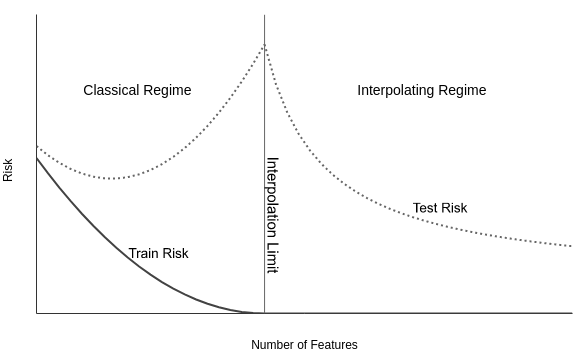}
  \caption{Illustration of the double-descent phenomenon in Random Fourier Features.}
  \label{fig:double-descent}
\end{figure}

Double descent occurs across many random-feature constructions.
Although pseudoinverse computation becomes more expensive in the wide
regime, it also reveals how large feature sets can regularize the
estimator. generalization performance is closely tied to the
\emph{effective condition number} of \(X^\top X\): empirical and
theoretical analyses~\cite{hogg2023featureweighting} show that test
error peaks precisely when this condition number is largest at
interpolation. As the number of features increases, the effective
condition number typically decreases, explaining the improved test error
in the overparameterized regime.

Not all settings exhibit double descent. Asymptotic theories
\cite{bartlett2020benign} require specific spectral properties—a few
dominant singular values and many moderate ones. If these are violated
(for example, by including extremely high deterministic frequencies),
performance may degrade. Feature weighting~\cite{hogg2023featureweighting}
or bandwidth control through the RFF parameter~\(\sigma\) mitigate this
effect: small~\(\sigma\) leads to smoother representations, while
large~\(\sigma\) yields highly oscillatory features that require more
samples to generalize. Thus, \(\sigma\) acts as a bias--variance
hyperparameter.

Finally, the behavior of RFFs can also be understood through their
Fourier-analytic origin.  Let \(f(t)\) be a square-integrable function
with Fourier transform
\begin{align}
  \hat{f}(\omega)
    &= \frac{1}{\sqrt{2\pi}} \int f(t)\,e^{i\omega t}\,dt,\\
  f(t)
    &= \frac{1}{\sqrt{2\pi}} \int \hat{f}(\omega)\,e^{-i\omega t}\,d\omega
\end{align}
To connect this with kernels, note that a shift-invariant kernel
\(k(x_t - x_s)=g(t-s)\) is fully characterized by its lag-domain profile
\(g(\tau)\), which plays the role of a stationary ``signal’’ defined on
time-lags.  By Bochner’s theorem,
\[
g(\tau)=\int p(\omega)\,e^{i\omega\tau}\,d\omega
\]
where \(p(\omega)\) is the spectral density of the kernel. For the case of RBFNets,  $\omega\sim \mathcal{N}(0,\gamma^2)$. We note that in the general case of sampling from RFFs, the choice of $\gamma$ is a form of regularization, preventing unnecessary high-frequency oscillations in the resulting interpolants (see for example \cite{hogg2023featureweighting}).

Random
Fourier Features approximate this integral via Monte Carlo sampling.
Defining
\[
\varphi(x)=
\sqrt{\frac{2}{m}}\,
\bigl[
\cos(\omega_1^\top x + b_1),\,
\ldots,\,
\cos(\omega_m^\top x + b_m)
\bigr]^\top\]
with $\omega_j\sim p$ and $b_j\sim\mathrm{Unif}[0,2\pi]$.

Then
\[
k(x_t -x_s)\approx \varphi(x_t)^\top\varphi(x_s)
\]
so RFFs may be interpreted as a random discretization of the Fourier
representation of the kernel’s lag-domain signal.

\subsection{Online Learning and Non-stationarity}

Many real-world systems, such as financial, physiological, or environmental processes, are inherently \emph{non-stationary}.  
Exponentially Weighted RLS (EWRLS) adapts to such dynamics through a forgetting factor~$\lambda<1$ that down-weights older data, enabling the tracking of time-varying correlations and structural breaks.

This study is motivated by the need for efficient forecasting methods for non-stationary time series commonly encountered in financial markets.  
Classical tools for analysing and decomposing non-stationary signals, widely used in communications, acoustics, and signal processing, include:

\begin{itemize}
  \item \textbf{Short-Time Fourier Transform (STFT):} a windowed Fourier transform,
  \[
  \hat{f}_{\lambda}(t,\omega)
    = \frac{1}{\sqrt{2\pi}} \int_{-\infty}^{\infty}
      \phi_{\lambda}(t-s) f(s) e^{-i\omega s}\,ds,
  \]
  where the window function $\phi_{\lambda}(r) = \phi(\lambda r)$ is often Gaussian with standard deviation~$1/\lambda$, although one-sided exponential kernels $\phi_{\lambda}(r)=e^{-\lambda r}$ for $r>0$ are also feasible ~\cite{TOMAZIC1996141,flandrin1998timefrequency}.  
  \item \textbf{Wavelet Transform:} or multiscale analysis,
  \[
  W_{\psi}[f](a,b)
    = \int_{-\infty}^{\infty} \psi^{*}\!\left(\frac{x-b}{a}\right) f(x)\,dx,
  \]
  where $\psi(\cdot)$ is the mother wavelet of compact support, and ${}^{*}$ denotes complex conjugation.
  
  \item \textbf{Wigner Transform:} or Wigner quasi-probability distribution,
  \[
  W_f(t,\omega)
    = \frac{1}{\sqrt{2\pi}} \int_{-\infty}^{\infty}
      f^{*}(t+s)\,f(t-s)\,e^{-i\omega s}\,ds
  \]
\end{itemize}

Each of these methods introduces localization in time or frequency, either through compactly supported basis functions or windowing. Among them, STFTs are arguably the most commonly used, and are used in the spectral analysis of sound and music (which clearly has an evolving wave-form), among other applications.  We are interested in the exponentially weighted STFT primarily because of its ease of update when new information is added.

\subsection{STFT and EWRLS with Random Fourier Features}To connect Random Fourier Features (RFFs) with short-term Fourier analysis, consider the exponentially weighted kernel term in our objective function:\begin{align*}e^{-\lambda (t-s)} k(x_t - x_s) &= e^{-\lambda\tau} \int p(\omega) e^{i\omega^\top (x_t - x_s)} d\omega\end{align*}where $\tau = t-s$ is the lag. The forgetting factor $e^{-\lambda\tau}$ acts as a one-sided \emph{exponential window} on the lag, modulating each Fourier component of the kernel in the same way as an exponentially weighted Short-Time Fourier Transform (STFT). Substituting this into the exponentially weighted objective yields the formulation:$$\hat{\beta}_t = \argmin_{\beta} \sum_{s \leq t} e^{-\lambda (t-s)} \bigl\| y_s - \varphi(x_s)^\top \beta \bigr\|^2$$
This problem admits the usual online solution via exponentially weighted recursive least-squares (EWRLS) updates. In this framework, the RLS algorithm behaves as a dynamic spectral estimator where the RFF weights $\hat{\beta}_t$ correspond to the instantaneous amplitudes of the signal's spectral components as resolved through the randomized basis.

\subsubsection{Computational Complexity and Spectral Deconvolution}
A critical divergence exists in the computational order of these approaches. Traditional block-based STFT implementations leverage the Fast Fourier Transform (FFT) to achieve a complexity of $O(M \log M)$ per frame, while recursive variants can reach $O(M)$ \cite{tomazic1996stft}. In contrast, the proposed ABO framework achieves a complexity of $O(ND)$ per update, where $N$ is the sliding window length and $D$ is the number of random features. 

Unlike standard covariance-form RLS which scales as $O(D^2)$ due to the maintenance and inversion of the $D \times D$ autocorrelation matrix \cite{sayed2003fundamentals, haykin2002adaptive}, the QR-based formulation in ABO leverages the structural properties of the sliding window to maintain a linear dependence on the feature dimension $D$. Specifically, by utilizing efficient updates for the minimum-norm least-squares solution in rank-deficient settings \cite{STAUB2021125996}, the update process is restricted to the $N$-dimensional subspace of active observations, resulting in $O(ND)$ operations. This allows the algorithm to scale effectively into the overparameterized regime required for benign overfitting.

\subsubsection{Extension Beyond the Time Domain}Unlike the STFT, which is mathematically confined to temporal or spatial dimensions, EWRLS-RFF generalizes spectral analysis to arbitrary high-dimensional feature spaces. Because RFFs provide an explicit mapping for any shift-invariant kernel, this adaptive framework can be applied to abstract inputs such as graph signals or multidimensional sensor arrays. In these regimes, the algorithm tracks localized structures in feature spaces where traditional grid-based FFT methods are fundamentally inapplicable.

\subsection{Stability through QR Updates}

Standard RLS methods using rank-one updates and downdates are known to be numerically unstable~\cite{sayed2003fundamentals,haykin2002adaptive}, depending on forgetting factors. In the case of overparameterized models, the updates and downdates have to be 
altered to consider the singular case (i.e., rank-increasing updates and rank-decreasing downdates). We include the details of update and downdate formulas in Appendix~\ref{appendix:sherman_morrison_version}, just for completeness. Nonetheless, 
in the overparameterized regime, this rank-one formulation of RLS also suffers from numerical instability. This fragility of covariance-form RLS motivates reformulating the recursion in terms of orthogonal-triangular (QR) updates, which maintain numerical stability under sequential adaptation~\cite{apolinario2009qrrls,higham2002accuracy}.

The inverse covariance, or the Moore--Penrose pseudoinverse in the underdetermined regime where \( D < N \), of the Gram matrix \( (X^\top X)^{\dagger} \) is updated recursively using the Sherman--Morrison formula. Although efficient, this formulation is numerically fragile: each rank-one correction involves a subtraction between nearly collinear outer products, which can destroy positive semi-definiteness and effectively double the condition number of the covariance estimate at each iteration~\cite{higham2002accuracy,hammarling2008updating}.

As a result, the covariance matrix may drift away from symmetry or become indefinite under finite-precision arithmetic, particularly in ill-conditioned or overparameterized settings.  
Square-root or QR-based variants address these issues by maintaining an orthogonal and triangular factorization updated through Givens or Householder rotations~\cite{alexander1993method,hammarling2008updating}, ensuring that round-off errors remain bounded.

\subsection{Contributions}
Building on these insights, we propose \emph{Adaptive Benign Overfitting (ABO)}, an adaptive, overparameterized EWRLS framework that integrates Random Fourier Features with exponentially weighted QR updates. Our contributions are twofold:
    \begin{itemize}
    \item We propose a numerically stable recursive scheme using Givens rotations for rank-one updates and downdates. By explicitly maintaining the Moore-Penrose pseudoinverse, the algorithm remains well-defined in rank-deficient regimes, ensuring the minimum-norm interpolant in overparameterized settings.
    \item Empirical evaluations on nonlinear synthetic, financial, and energy time-series demonstrate that the method achieves bounded residuals and numerical stability while successfully recovering the double-descent behavior characteristic of benign overfitting.
    \item  We show that ABO achieves $O(ND)$ complexity per update by leveraging the sliding-window structure and pseudoinverse properties. This linear scaling in feature dimension $D$ (where $D \gg N$) provides a significant improvement over the $O(D^2)$ cost of standard RLS \cite{haykin2002adaptive}, enabling high-dimensional online learning.
\end{itemize}

This formulation bridges the gap between theoretical analyses of benign overfitting and practical recursive algorithms for online learning, offering a unified framework for adaptive modeling in non-stationary environments. We now turn to the formal model and algorithmic formulation, extending Recursive Least Squares (RLS) to the overparameterized and rank-deficient regime while preserving numerical stability under sequential updates.

\subsection{Related Work}

Classical Recursive Least Squares (RLS) algorithms have been extensively studied in adaptive filtering and signal processing \cite{haykin2002adaptive,sayed2003fundamentals}.  
The introduction of exponential weighting improved their ability to track non-stationary signals \cite{apolinario2009qrrls}, although the standard covariance formulation remains prone to numerical instability due to repeated matrix updates, even when using standard ridge penalties\footnote{Standard ridge penalties in EWRLS are typically vanishingly small when the sample size grows, due to ease of update. Consequently their stabilization properties are somewhat limited.}.  
Square-root and QR-based variants were later developed to enhance numerical robustness \cite{alexander1993method,hammarling2008updating,higham2002accuracy}, yet these analyses typically assume full-rank and well-conditioned data.

Addressing the ill-conditioning inherent in high-dimensional adaptive filtering has traditionally relied on explicit regularization, most notably through sparsity constraints. Inspired by the LASSO, Zero-Attracting RLS (ZA-RLS) algorithms and their reweighted variants impose an $\ell_1$-norm penalty (or approximate $\ell_0$ or {\sl counting} ``norm'' penalty)  on the weight vector $w_n$, effectively forcing 'insignificant' coefficients to zero to reduce variance \cite{2011EksiogluTanc,2017HongGaoChen}. Similar stability guarantees have been achieved using projection-based methods, such as the Adaptive Projected Subgradient Method (APSM), which constrains weights to a convex set \cite{YAMADA2004639}.

In contrast to other regularization methods for adaptive filters, our work investigates implicit regularization arising from overparameterization, or {\sl benign overfitting}, \cite{belkin2019reconciling,bartlett2020benign,hastie2022surprises}.  
Rather than artificially constraining the model capacity via sparsity or projections, we expand the feature space (via RFFs, see \cite{rahimi2007random,mei2022generalization,tsigler2020benign,su2023benign} ) beyond the interpolation threshold. We note that while little work has been done in the rank-deficient regime, there has been some examination of both Cholesky and rank-factorization methods in this area \cite{Sentana1999,STAUB2021125996}.

As suggested by recent findings in the 'benign overfitting' literature, our approach allows the 'noise' energy to be distributed harmlessly across orthogonal directions in the high-dimensional space, stabilizing the RLS update without requiring the sparse structural assumptions of ZA-RLS.

The present work bridges these domains by formulating a QR-based exponentially weighted RLS algorithm that remains numerically stable in the overparameterized regime while preserving the recursive structure required for real-time applications.

\newpage

\section{System Model}
\label{sec:system_model}
The limitations of the covariance-updating approach discussed in
Appendix~\ref{appendix:sherman_morrison_version} highlight the need for a numerically
stable recursive formulation that operates reliably in both full-rank and
rank-deficient regimes.  To this end, we reformulate the RLS problem using
an orthogonal–triangular (QR) factorization, which avoids explicit
inversion of Gram matrices and preserves numerical stability under sliding
window updates.

At each time step $t$ we observe a raw covariate $x_t \in \mathbb{R}^d$
and a scalar response $y_t \in \mathbb{R}$.  Random Fourier Features
provide a nonlinear mapping
\[
    z_t = \varphi(x_t) \in \mathbb{R}^D ,
\]
and the regression is performed on the RFF design matrix.  Over a rolling
window of length $N$, the design matrix and response vector are
\[
    Z_t =
    \begin{bmatrix}
        z_{t-N+1}^\top \\
        \vdots \\
        z_t^\top
    \end{bmatrix}
    \in \mathbb{R}^{N \times D}, 
    \qquad
    Y_t =
    \begin{bmatrix}
        y_{t-N+1} \\
        \vdots \\
        y_t
    \end{bmatrix}
    \in \mathbb{R}^N .
\]

Classical least squares requires forming the Gram matrix
$A_t = Z_t^\top Z_t$, which may be singular when $D > N$ and is
numerically unstable to update rank-one.  Instead, QR–RLS maintains an
orthogonal–triangular factorization \label{sec:System}
\begin{equation}
    Z_t = Q_t R_t ,
    \label{eq:qr-factorization}
\end{equation}
where $Q_t \in \mathbb{R}^{N \times N}$ is orthogonal and
$R_t \in \mathbb{R}^{N \times D}$ is upper–trapezoidal.
As new data $(z_{t+1}, y_{t+1})$ arrive, and the oldest row is removed,
the factors $(Q_t, R_t)$ are updated using Givens rotations, avoiding
explicit matrix inversion and ensuring numerical stability in both
full-rank and rank-deficient regimes.

\subsection{Least squares via QR}

The least-squares estimate solves
\[
    \min_{\beta}\;\|Z_t \beta - Y_t\|_2^2 
\]
Using the QR factorization \eqref{eq:qr-factorization} and the
orthogonality of \(Q_t\), the objective becomes
\begin{align}
    \|Z_t \beta - Y_t\|_2^2
        &= \|Q_t^\top(Z_t \beta - Y_t)\|_2^2 \\
        &= \|R_t \beta - Q_t^\top Y_t\|_2^2 
        \label{eq:ls_qr_objective}
\end{align}
Thus the LS problem reduces to solving the triangular system
\[
    R_t \beta = Q_t^\top Y_t 
\]

When \(R_t\) has full column rank (\(D \leq N\)), the solution is obtained
by back-substitution:
\begin{equation}   
    \beta_t = R_t^{-1} Q_t^\top Y_t 
\end{equation}
In the overparameterized regime (\(D > N\)), \(R_t\) is rank-deficient and
the minimum-norm solution is obtained by replacing \(R_t^{-1}\) with its
Moore--Penrose pseudoinverse: 
\begin{equation}
    \beta_t = R_t^{\dagger} Q_t^\top Y_t 
       \label{eq:min_norm_soln}
\end{equation}

This expression holds in both regimes and avoids forming the Gram matrix
\(Z_t^\top Z_t\), ensuring numerical stability.

\newpage

\noindent\textbf{Remark 1 (Minimum-Norm and Benign Overfitting).}\label{rem:benign_overfitting} In the overparameterized regime ($D > N$), Equation~\eqref{eq:min_norm_soln} yields the unique minimum-$\ell_2$-norm interpolant:\begin{equation}\beta_t = \arg\min_{\beta} \|\beta\|_2 \quad \text{s.t.} \quad Z_t \beta = Y_t,\end{equation}equivalent to the Moore-Penrose solution $\beta_t = Z_t^{\dagger} Y_t$. In this setting, generalization is governed by the implicit regularization of the minimum-norm constraint rather than explicit shrinkage. As the null space of $Z_t$ expands, the estimator distributes variance across many tail eigenvalues—a spectral property essential for \emph{benign overfitting}. This mechanism is consistent with the double-descent behavior observed in online learning and justifies the use of stable recursive updates for the pseudoinverse solution.

\subsection{Random Fourier Features}

To introduce nonlinearity while retaining a tractable linear model, we map
each raw observation \(x_t \in \mathbb{R}^d\) into a higher-dimensional
random feature space using \emph{Random Fourier Features} (RFFs)
\cite{rahimi2007random}.

Let \(A \in \mathbb{R}^{d \times D}\) have columns
\(\omega_1,\dots,\omega_D\) drawn independently from \(p(\omega)\),
and let \(b \in \mathbb{R}^D\) contain i.i.d.\ phases sampled from
\(\mathrm{Unif}[0,2\pi]\).  The random feature mapping is
\begin{equation}
    z_t
    = z(x_t)
    = \sqrt{\frac{2}{D}}
      \cos(A^\top x_t + b)
    \in \mathbb{R}^D 
    \label{eq:rff}
\end{equation}
where the cosine is applied elementwise. We adopt this formulation for notational compactness, although many would include $\cos$ and $\sin$, or adopt novel sampling schemes to ensure more efficient representations.\footnote{Improved efficiency is known to be found via {\sl Orthogonal Random Features, Quasi-Monte Carlo Random Features, Structured Orthogonal Random Features and Fast Food} based sampling methods, among others. }

Over a rolling window of length \(N\), the RFF design matrix is
\[
    Z_t
    =
    \begin{bmatrix}
        z_{t-N+1}^\top \\
        \vdots \\
        z_t^\top
    \end{bmatrix}
    \in \mathbb{R}^{N \times D}
\]
which forms the input to the QR--based recursive least-squares algorithm
described in the following sections.  The dimensionality \(D\) is fixed
and chosen independently of the sample size, enabling efficient updates in
both full-rank and overparameterized regimes.

\subsection{QR-RLS Update via Givens Rotations}

\subsubsection{Givens Rotations}

A Givens (plane) rotation \(G(i,j;c,s)\in\mathbb{R}^{N\times N}\) is the
identity matrix except on rows and columns \(i\) and \(j\), where
\[
G(i,j;c,s)_{[i,j],[i,j]}=
\begin{bmatrix}
c & s\\
-s & c
\end{bmatrix},
\qquad c^2 + s^2 = 1
\]
Given a pivot \(x\) and a subdiagonal entry \(y\) in the same column, we
choose
\[
    r = \sqrt{x^2 + y^2}, \qquad
    c = \frac{x}{r}, \qquad
    s = \frac{y}{r}
\]
so that
\[
\begin{bmatrix}
c & s\\ -s & c
\end{bmatrix}
\begin{bmatrix}
x\\ y
\end{bmatrix}
=
\begin{bmatrix}
r\\ 0
\end{bmatrix}
\]
A sequence of such rotations \(\{G_\ell\}\) is applied to eliminate the
subdiagonal entries of the augmented matrix, and their product forms the
orthogonal factor:
\[
    Q_t^\top = G_1^\top G_2^\top \cdots G_k^\top 
\]
This yields the QR relation
\[
    Q_t^\top Z_t = R_t
\]
where \(R_t\) is the updated upper–trapezoidal factor.

\subsubsection{Observation Update}

When a new feature vector \(z_t^\top\) arrives, we first form the
augmented matrix
\[
    \bar{R}_t
    =
    \begin{bmatrix}
        R_{t-1} \\
        z_t^\top
    \end{bmatrix}
    \in \mathbb{R}^{(N+1)\times D}
\]
which corresponds to appending the new row to the windowed design matrix
\(Z_{t-1}\).  This matrix is no longer upper–trapezoidal: its last row
contains subdiagonal entries that must be eliminated.

A sequence of Givens rotations \(G_t\) is applied to zero these
subdiagonal elements:
\begin{equation}
\label{eq:givens_update}
    G_t^\top 
    \begin{bmatrix}
        R_{t-1} \\
        z_t^\top
    \end{bmatrix}
    =
    \begin{bmatrix}
        R_t \\
        0^\top
    \end{bmatrix}
\end{equation}
yielding the updated triangular factor \(R_t \in \mathbb{R}^{N\times D}\).
The corresponding orthogonal factor updates as
\[
    Q_t
    =
    \begin{bmatrix}
        Q_{t-1} & 0 \\
        0 & 1
    \end{bmatrix}
    G_t 
\]
so that the augmented design matrix satisfies
\[
    \begin{bmatrix}
        Z_{t-1} \\
        z_t^\top
    \end{bmatrix}
    = Q_t
      \begin{bmatrix}
          R_t \\
          0^\top
      \end{bmatrix}
\]

The least–squares estimate is then
\[
    \beta_t = R_t^{\dagger} Q_t^\top Y_t 
\]
Instead of recomputing \(R_t^{\dagger}\) from scratch, the pseudoinverse is
updated through the Givens transformations:
\[
    R_t^{\dagger}
    =
    \left( G_t^\top
    \begin{bmatrix}
        R_{t-1} \\
        z_t^\top
    \end{bmatrix}
    \right)^{\dagger}
    =
    \begin{bmatrix}
        R_{t-1} \\
        z_t^\top
    \end{bmatrix}^{\dagger}
    G_t 
\]
which follows from \((GA)^\dagger = A^\dagger G^\top\) for any orthogonal
matrix \(G\).  This recursion avoids explicit inversion and preserves
numerical stability during online updates.

\subsubsection{Pseudoinverse Update (Greville/Cline) and Link to QR-RLS}

From the previous section we have the row-augmented matrix $\bar{R}_t$,
together with its Moore--Penrose pseudoinverse $\bar{R}_t^\dagger$.

Define:
\begin{align*}
c &= (I - R_{t-1}^\dagger R_{t-1})\, z_t
\\
h &= z_t^\top R_{t-1}^\dagger 
\\
c^\dagger &= \dfrac{c}{\|c\|_2^2}
\end{align*}
Then the Greville/Cline update for the pseudoinverse is
\begin{equation}
\label{eq:pseudoinv_update}
\bar{R}_t^\dagger =
\begin{bmatrix}
R_{t-1}^\dagger - b h & b
\end{bmatrix},
\qquad
b =
\begin{cases}
\hspace{.5cm}c^\dagger & \text{if } c \neq 0,\\[1.25ex]
\dfrac{R_{t-1}^\dagger h^\top}{1 + h h^\top} & \text{if } c = 0
\end{cases}
\end{equation}
which covers the full-column-space and dependent-column cases \cite{cline1964representations}.
Unlike standard QRD-RLS, which is typically restricted to overdetermined systems ($D \le N$), 
ABO utilizes Greville/Cline updates to maintain a stable minimum-norm solution in the $D > N$ regime.

Define the transformed right-hand side as
\begin{equation}
    w_t = Q_t^\top Y_t
\label{eq:transformed_rhs}
\end{equation}

which can be updated recursively via
\begin{equation*}
    w_t = G_t^\top
    \begin{bmatrix}
        w_{t-1} \\ y_t
    \end{bmatrix}
\end{equation*}

By orthogonality of $Q_t$, we have
\begin{equation*}
    \big(G_t^\top 
    \begin{bmatrix}
        R_{t-1}\\
        z^\top_t
    \end{bmatrix}\big)^{\dagger}
    = R_t^{\dagger}
    \;\;\Longleftrightarrow\;\;
    \begin{bmatrix}
        R_{t-1}\\
        z^\top_t
    \end{bmatrix}^{\dagger} G_t
    = R_t^{\dagger}
\end{equation*}
    
Hence, we apply $G_t$ to both the triangular factor and the transformed 
right-hand side. By orthogonality, $G_t^\top G_t = I$, so
\begin{align*}
    \beta_t 
    &= \begin{bmatrix}
        R_{t-1} \\ z^\top_t
      \end{bmatrix}^{\dagger} G_t G_t^\top 
      \begin{bmatrix}
        w_{t-1} \\ y_t
      \end{bmatrix} \\
    &= \begin{bmatrix}
        R_{t-1} \\ z^\top_t
      \end{bmatrix}^{\dagger} 
      \begin{bmatrix}
        w_{t-1} \\ y_{t}
      \end{bmatrix}
\end{align*}

\subsubsection{Weight Update}

\begin{align}
\label{eq:weights_update}
    \beta_{t} &= R_{t}^\dagger Q_{t}^\top y_{t} \\
     &= \begin{bmatrix}
        R_{t-1}^\dagger - b z^\top_t R_{t-1}^\dagger  & b
    \end{bmatrix} G_t G_t^\top  \begin{bmatrix}
        w_{t-1} \\
        y_t
    \end{bmatrix} \\
    &= R_{t-1}^\dagger w_{t-1} - b z^\top_t R_{t-1}^\dagger w_{t-1} + b y_t \\
    &= \beta_{t-1} + b(y_t - z^\top_t\beta_{t-1})
\end{align}

The vector $b$ plays the role
of the Kalman gain; its expression depends on whether the new row $z_t$
increases the rank of $R_{t-1}$ ($c \neq 0$) or lies in its column
space ($c = 0$).

\subsubsection{Observation Downdate}\label{sec:downdate}

To enforce a rolling window, the contribution of the oldest observation
must be removed from the QR factors.  Let $G_{t-N}$ denote the Givens
rotation that reverses the orthogonal transformation previously applied to
the first row of $Q_{t-1}$ \cite{golub1996matrix}.  Partition
$Q_{t-1}$ as
\[
Q_{t-1} =
\begin{bmatrix}
    q_{t-N}^\top \\
    Q_{t-1}^{(2:N)}
\end{bmatrix}
\]
where $q_{t-N}^\top$ is the row corresponding to the discarded
observation.  The downdate rotation $G_{t-N}$ is chosen so that
\[
G_{t-N}^\top q_{t-N} = \alpha e_1,
\qquad \alpha \in \{\pm1\}
\]
which yields
\[
Q_{t-1} G_{t-N}
=
\begin{bmatrix}
    \alpha & 0 \\
    0      & Q_t
\end{bmatrix}
\]
If $\alpha=-1$, we flip the sign of $G_{t-N}$ to preserve orientation.

Applying $G_{t-N}^\top$ to the triangular factor restores the first row of
$R_{t-1}$ to its pre-rotation value $z_{t-N}^\top$.  Hence
\[
G_{t-N}^\top R_{t-1}
=
\begin{bmatrix}
    z_{t-N}^\top \\
    R_t
\end{bmatrix}
\]
Subtracting the rank-one term $e_1 z_{t-N}^\top$ removes this restored row:
\[
G_{t-N}^\top R_{t-1} - e_1 z_{t-N}^\top
=
\begin{bmatrix}
    0 \\ R_t
\end{bmatrix}
\]
After discarding the zero row, the remaining block is the downdated
triangular factor $R_t$.

For later use, define the intermediate matrix
\[
\hat{R}_t := R_{t-1} - G_{t-N} e_1 z_{t-N}^\top
\]
so that
\[
G_{t-N}^\top \hat{R}_t = 
\begin{bmatrix}
    0 \\ R_t
\end{bmatrix}
\]

Since $G_{t-N}$ is orthogonal, the generalized inverse identity
$(GA)^\dagger = A^\dagger G^\top$ applies \cite{cami1979generalized}, giving
\[
(G_{t-N}^\top \hat{R}_t)^\dagger
= \hat{R}_t^\dagger G_{t-N}
\]
This relation will be used to express the downdated pseudoinverse and to
connect the downdate with the subsequent weight update.

\subsubsection{Pseudoinverse Downdate via Generalized Sum Formulas}

To apply the generalized inverse sum formulas of
\cite{cami1979generalized,cline1964representations}, introduce
\begin{align*}
    h &= z_{t-N}^\top R_{t-1}^{\dagger}, \qquad
    k = R_{t-1}^{\dagger} G_{t-N} e_1,\\
    u &= (I - R_{t-1} R_{t-1}^{\dagger})G_{t-N} e_1, \qquad
    v = z_{t-N}^\top (I - R_{t-1}^{\dagger} R_{t-1})
\end{align*}
For a downdate, the removed row $z_{t-N}^\top$ lies in both the row and
column spaces of $R_{t-1}$, and therefore
\[
u = 0, \qquad v = 0
\]
\textbf{Singular case}:
\begin{equation}
    \label{eq:pseudoinv_downdate}
\hat{R}_t^\dagger 
= R_{t-1}^{\dagger}
  - k k^\dagger R_{t-1}^{\dagger}
  - R_{t-1}^{\dagger} h^\dagger h
  + (k^\dagger R_{t-1}^{\dagger} h^\dagger)\, k h
\end{equation}
\textbf{Nonsingular case}:
\begin{equation}
\label{eq:inv_downdate}
\hat{R}_t^\dagger
= R_{t-1}^{-1}
  + \frac{
      R_{t-1}^{-1} G_{t-N} e_1 z_{t-N}^\top R_{t-1}^{-1}
    }{
      1 - z_{t-N}^\top R_{t-1}^{-1} G_{t-N} e_1
    }
\end{equation}

In both cases, the downdated pseudoinverse $R_t^\dagger$ is obtained by
right-multiplying by the orthogonal rotation:
\[
\begin{bmatrix}
    0 & R_t^\dagger
\end{bmatrix}
=
\hat{R}_t^\dagger G_{t-N}
\]
This yields a fully recursive pseudoinverse downdate compatible with both
full-rank and rank-deficient regimes.

\subsubsection{Weight Downdate}

To compute the iterative weight downdate, we begin by noting that the update of $Y_{t-1}$ requires removing its first element, analogously to the removal of the first observation $z_{t-N}$ in the feature matrix. Since $y_{t-N}$ is a scalar, this operation can be expressed as
\begin{equation*}
    G_{t-N}^\top Q_{t-1}^\top
    \begin{bmatrix}
        y_{t-N} \\[0.3em] Y_{t}
    \end{bmatrix}
    -
    \begin{bmatrix}
        y_{t-N} \\[0.3em] 0
    \end{bmatrix}
    =
    \begin{bmatrix}
        0 \\[0.3em] Q_t^\top Y_t
    \end{bmatrix}
\end{equation*}
where $y_{t-N} = e_1^\top Y_{t-1}$ denotes the first entry of the vector $Y_{t-1}$.  
This relation can be rearranged as


\begin{equation*}
    \begin{bmatrix}
    0 \\ Q_t^\top Y_t
\end{bmatrix} = G_{t-N}^\top Q_{t-1}^\top Y_{t-1} - e_1 y_{t-N}
\end{equation*}

Consequently, the complete downdating expression for the regression weights becomes
\begin{align*}
    \beta_t 
    &= R_t^{\dagger} Q_t^\top Y_t \\
    &= \bigl(R_{t-1} - G_{t-N} e_1 z_{t-N}^\top\bigr)^{\dagger} 
       G_{t-N} \bigl(G_{t-N}^\top Q_{t-1}^\top Y_{t-1} - e_1 y_{t-N}\bigr) \\
    &= \bigl(R_{t-1} - G_{t-N} e_1 z_{t-N}^\top\bigr)^{\dagger} 
       \bigl(Q_{t-1}^\top Y_{t-1} - G_{t-N} e_1 y_{t-N}\bigr)
\end{align*}

In the final step, the appropriate pseudoinverse or inverse update formula is substituted, depending on the regime in which the system operates—that is, whether the feature matrix is overdetermined or underdetermined.

\hspace{-.6cm}
\textbf{Singular case}:

To apply the generalized pseudoinverse downdate identity, we require the
standard range conditions \cite{cami1979generalized}:
\[
    G_{t-N} e_1 \in \operatorname{Range}(R_{t-1}), 
    \qquad
    z_{t-N} \in \operatorname{Range}(R_{t-1}^\top)
\]
Since both $G_{t-N} e_1$ and $z_{t-N}$ are vectors, these conditions ensure
that the corresponding oblique projections vanish:
\[
    (I - R_{t-1} R_{t-1}^{\dagger})\, G_{t-N} e_1 = 0,
    \qquad
    z_{t-N}^\top (I - R_{t-1}^{\dagger} R_{t-1}) = 0^\top
\]

Under these assumptions, the rank of the augmented system is preserved:
\[
    \rho(R_{t-1}) = \rho(\hat R_{t})
\]
so the downdate does not alter the column space of $R_{t-1}$.

Consider the block matrix
\[
    \hat R_{t-1} =
    \begin{bmatrix}
        R_{t-1} & G_{t-N}e_1 \\[0.3em]
        z_{t-N}^\top & D
    \end{bmatrix}
\]
with Schur complement $D = G_{t-N}e_1 R_{t-1}^\dagger z_{t-N}^\top$.  
Following \cite{hartwig1976rankfactorization}, we set $D = 1$ to avoid
introducing additional scale factors.  
This directly yields the normalization
\[
    1 = z_{t-N}^\top  R_{t-1}^{\dagger} G_{t-N} e_1
\]

Because $k k^\dagger k = k$, the matrix $k k^\dagger$ acts as the
orthogonal projector onto the subspace spanned by the vector \(k = R_{t-1}^{\dagger} G_{t-N} e_1\).

Since the model interpolates the data within the window, the response vector
$Y_{t-1}$ lies in the row space of $Z_{t-1}$.  
Equivalently, the estimator $\beta_{t-1}$ satisfies
\[
    z_{t-N}^\top \beta_{t-1} = y_{t-N}
\]

We now derive the downdated coefficient vector $\beta_t$.

\begin{equation}
\label{eq:beta_pseudoinv_downdate}
\begin{split}
\beta_t &= R^\dagger_t Q_t^\top Y_t \\ &= (R_{t-1}^{\dagger} - kk^\dagger R_{t-1}^{\dagger} - R_{t-1}^{\dagger} h^\dagger h + (k^\dagger R_{t-1}^{\dagger} h^\dagger)\,kh) \\ &\hspace{.5cm}(Q_{t-1}^\top Y_{t-1} - G_{t-N} e_1 y_{t-N}) \\ &= \beta_{t-1} - kk^\dagger \beta_{t-1} - R_{t-1}^{\dagger} h^\dagger z_{t-N}^\top \beta_{t-1} \\ & \hspace{.5cm} + (k^\dagger R_{t-1}^{\dagger} h^\dagger) k z_{t-N}^\top \beta_{t-1} \\ & \hspace{.5cm} - k y_{t-N} + kk^\dagger k y_{t-N} \\ & \hspace{.5cm} + R_{t-1}^{\dagger} h^\dagger z_{t-N}^\top R_{t-1}^{\dagger} G_{t-N} e_1 y_{t-N} \\ &\hspace{.5cm} - (k^\dagger R_{t-1}^{\dagger} h^\dagger) k z_{t-N}^\top R_{t-1}^{\dagger} G_{t-N} e_1 y_{t-N} \\ &= \beta_{t-1} - kk^\dagger \beta_{t-1} \\ & \hspace{.5cm} + ( R_{t-1}^{\dagger} h^\dagger - (k^\dagger R_{t-1}^{\dagger} h^\dagger) k)(y_{t-N} - z_{t-N}^\top \beta_{t-1}) \\ & \hspace{.5cm} - k y_{t-N} + k y_{t-N} \\ \beta_t &= \beta_{t-1} - kk^\dagger \beta_{t-1} 
\end{split}
\end{equation}

The update is therefore the orthogonal projection of $\beta_{t-1}$ onto the
subspace orthogonal to $k$:
\[
    \beta_t = (I - k k^\dagger)\,\beta_{t-1}
\]
Since $k k^\dagger$ is the projector onto $\operatorname{Range}(k)$, the
downdated parameter satisfies
\[
    k^\top \beta_t = 0
\]
ensuring that information associated with the discarded row is completely
removed from the estimate.

\hspace{-.6cm}
\textbf{Non Singular case}:
\begin{equation}
\label{eq:beta_inv_downdate}
\begin{split}
    \beta_t &= R^\dagger_t Q_{t}^\top Y_t 
     \\ 
     &= (R^\dagger_{t-1} + \frac{R^\dagger_{t-1} G_{t-N} e_1 z_{t-N}^\top R^\dagger_{t-1}}{(1 - z_{t-N}^\top R^\dagger_{t-1} G_{t-N} e_1)}) \\ &\hspace{.5cm}(Q_{t-1}^\top Y_{t-1} - G_{t-N} e_1 y_{t-N})
     \\
     &= \beta_{t-1} + \frac{R^\dagger_{t-1} G_{t-N} e_1 z_{t-N}^\top \beta_{t-1}}{(1 - z_{t-N}^\top R^\dagger_{t-1} G_{t-N} e_1)}  \\ & \hspace{.5cm} - y_{t-N} R^\dagger_{t-1} G_{t-N} e_1 (1 + \frac{z_{t-N}^\top R^\dagger_{t-1} G_{t-N} e_1}{(1 - z_{t-N}^\top R^\dagger_{t-1} G_{t-N} e_1)})  
     \\
     &= \beta_{t-1} + \frac{R^\dagger_{t-1} G_{t-N} e_1 z_{t-N}^\top \beta_{t-1}}{(1 - z_{t-N}^\top R^\dagger_{t-1} G_{t-N} e_1)} \\ &\hspace{.5cm} - \frac{y_{t-N} R^\dagger_{t-1} G_{t-N} e_1}{(1 - z_{t-N}^\top R^\dagger_{t-1} G_{t-N} e_1)} 
     \\
     \beta_{t} &= \beta_{t-1} - \frac{R^\dagger_{t-1} G_{t-N} e_1}{(1 - z_{t-N}^\top R^\dagger_{t-1} G_{t-N} e_1)} (y_{t-N} - z_{t-N}^\top \beta_{t-1})
\end{split}
\end{equation}

In this formulation, the Kalman gain occupies its standard position as the multiplicative factor of the innovation term $(y_{t-N} - z_{t-N}^\top \beta_{t-1})$, thereby updating the estimate in proportion to the prediction error.

The overall recursive procedure is summarized in Algorithm~\ref{alg:abo_qr}, which integrates
the exponentially weighted updates and QR-based orthogonalization described above.

\subsubsection{Forgetting factor and Regularization}\label{sec:forgetting_factor}

The introduction of a forgetting factor $\lambda \in (0,1]$ in the RLS cost function allows the algorithm to adapt to non-stationary signals by exponentially down-weighting past observations. This mechanism also acts as an implicit form of regularization, improving the numerical conditioning of the covariance matrix \cite{sayed2003fundamentals,apolinario2009qrrls}.

The previous derivations correspond to the special case $\lambda = 1$, where all samples are equally weighted. In practice, however, the forgetting factor is incorporated through a diagonal weighting matrix
\[
\Lambda_t = \operatorname{diag}(\sqrt{\lambda^{N-1}}, \sqrt{\lambda^{N-2}}, \ldots, \sqrt{\lambda^{0}}),
\]
leading to the exponentially weighted least-squares problem
\[
\min_{\beta} \, \| \Lambda_t (Z_t \beta_t - Y_t) \|_2^2.
\]
which using the QR factorization it is equal to
\[
\beta = R_{t-1}^\dagger Q_t^\top \Lambda_t Y_t
\]

To obtain a recursive QR update with exponential forgetting, assume that at
time $t-1$
\[
\Lambda_{t-1} Z_{t-1} = Q_{t-1} R_{t-1}
\]
Incorporating a new observation $z_t$ with forgetting factor
$\lambda \in (0,1]$ yields
\[
\Lambda_t Z_t
=
\begin{bmatrix}
\sqrt{\lambda}\,\Lambda_{t-1} & 0 \\
0 & 1
\end{bmatrix}
\begin{bmatrix}
Z_{t-1} \\
z_t^\top
\end{bmatrix}
=
\begin{bmatrix}
\sqrt{\lambda}\, Q_{t-1} R_{t-1} \\
z_t^\top
\end{bmatrix}
\]
Factoring out the orthogonal block $\operatorname{diag}(Q_{t-1},1)$ gives
\begin{equation}
\Lambda_t Z_t
=
\begin{bmatrix}
Q_{t-1} & 0 \\
0 & 1
\end{bmatrix}
\begin{bmatrix}
\sqrt{\lambda}\, R_{t-1} \\
z_t^\top
\end{bmatrix}
\label{eq:qr_recursive}
\end{equation}
so the update reduces to computing the QR factorization of the stacked
matrix in~\eqref{eq:qr_recursive}.

The weighted response vector updates analogously as
\[
\Lambda_t Y_t
=
\begin{bmatrix}
\sqrt{\lambda}\,\Lambda_{t-1} Y_{t-1} \\
y_t
\end{bmatrix}
\]

Accordingly, the recursive factorization update becomes
\[
R_t =  
\begin{bmatrix}
    \sqrt{\lambda}\, R_{t-1} \\
    z_t^\top
\end{bmatrix}
\]
When updating the pseudoinverse matrix, the forgetting factor is absorbed by scaling
\[
R_{t-1}^\dagger \leftarrow \frac{1}{\sqrt{\lambda}} R_{t-1}^\dagger
\]
prior to performing the downdate step in~\eqref{eq:pseudoinv_update}.  
This adjustment ensures that the inverse recursively reflects the exponential weighting. A detailed proof is provided in Appendix~\ref{appendix:forgetting_factor_proof}.

Regularization can also be introduced through \emph{diagonal loading}, which modifies the normal equations to
\[
(Z^\top Z + \delta I)\beta = Z^\top Y
\]
thereby converting a singular or ill-conditioned system into a well-posed one \cite{principe2011kernel,haykin2002adaptive}. 
While diagonal loading (ridge) provides stability by shifting the spectrum, ABO achieves stability through the orthogonal decomposition itself, allowing the model to utilize the full null space for variance absorption without the bias introduced by $\delta I$. In our setting, we deliberately avoid such explicit regularization, since our goal is to study over-parameterized regimes where the model is intentionally allowed to interpolate the data, which is entirely disallowed in the case of ridge-regularization.

\subsection{QR--EWRLS Algorithm}
\label{subsec:algorithm}

We now summarize the proposed QR-based exponentially weighted recursive least-squares
(QR--EWRLS) procedure. Algorithm~\ref{alg:abo_qr} presents the complete update cycle.
The following theorem establishes that, under idealized arithmetic, the algorithm computes
the minimum-norm exponentially weighted sliding-window least-squares solution.

\newpage
\captionsetup{type=algorithm}
\captionof{algorithm}{Adaptive Benign Overfitting (ABO) via QR-EWRLS}
\label{alg:abo_qr}
\footnotesize
\begin{algorithmic}[1]
\Require Stream $\{(x_t, y_t)\}$, window size $N$, forgetting factor $\lambda$, RFF dimension $D$ 
\Ensure Online predictions $\{\hat{y}_t\}$, coefficient vectors $\{\beta_t\}$ 

\State \textbf{Initialization:} 
\State \quad Compute initial $Z_0 \in \mathbb{R}^{N \times D}$ and $Y_0 \in \mathbb{R}^N$ 
\State \quad Perform QR factorization: $Z_0 = Q_0 R_0$ \Comment{[(\ref{eq:qr-factorization})]}
\State \quad Solve $R_0 \beta_0 = Q_0^\top Y_0$ for minimum-norm $\beta_0$ \Comment{[(\ref{eq:min_norm_soln})]}

\For{$t = N+1, N+2, \dots$}
    \State \textbf{Step 1: RFF Mapping}
    \State \quad $z_t \gets \sqrt{\frac{2}{D}}\cos(A^\top x_t + b)$ \Comment{[(\ref{eq:rff})]}
    
    \State \textbf{Step 2: Non-stationary Scaling}
    \State \quad $R_{t-1} \gets \sqrt{\lambda} R_{t-1}$, \quad $Y_{t-1} \gets \sqrt{\lambda} Y_{t-1}$ \Comment{[\S \ref{sec:forgetting_factor}]}
    
    \State \textbf{Step 3: Observation Update (Givens Rotations)} [(\ref{eq:givens_update})]
    \State \quad $\bar{R}_t \gets \begin{bmatrix} R_{t-1} \\ z_t^\top \end{bmatrix}$, \quad  $Y_t \gets \begin{bmatrix} Y_{t-1} \\ y_t \end{bmatrix}$ 
    \State \quad Apply $G_t$ such that $G_t^\top \bar{R}_t \gets \begin{bmatrix} R_t \\ \mathbf{0}^\top \end{bmatrix}$
    \State \quad $Q_t \gets \begin{bmatrix} Q_{t-1} & 0 \\ 0 & 1 \end{bmatrix} G_t$

    \State \textbf{Step 4: Pseudoinverse Update}
    \quad \State $h \gets z_t^\top R_{t-1}^\dagger$
    \If {$D \le N$} \Comment{Classical Regime} [(\ref{eq:pseudoinv_update})]
        \quad \State $\bar{R}_t^\dagger \gets
        \begin{bmatrix}
        R_{t-1}^\dagger - \dfrac{R_{t-1}^\dagger h^\top}{1 + h h^\top} h & \dfrac{R_{t-1}^\dagger h^\top}{1 + h h^\top}
        \end{bmatrix} 
        $ 
    \quad \Else \Comment{Overparameterized Regime} [(\ref{eq:pseudoinv_update})]
        \quad \State $\bar{R}_t^\dagger \gets
        \begin{bmatrix}
        R_{t-1}^\dagger - c^\dagger h & c^\dagger
        \end{bmatrix} $\Comment{Min-norm solution} 
    \quad \EndIf
    \quad \State $\bar{R}_t^\dagger G_t
    \gets R_t^{\dagger}$
    
    \State \textbf{Step 5: Regression Weights}
    \If {$D \le N$} \Comment{Classical Regime} [(\ref{eq:weights_update})]
        \quad \State  $\beta_t \gets \beta_{t-1} + \dfrac{R_{t-1}^\dagger h^\top}{1 + h h^\top}(y_t - z^\top_t\beta_{t-1})$ 
    \quad \Else \Comment{Overparameterized Regime} [(\ref{eq:weights_update})]
        \quad \State $\beta_t \gets \beta_{t-1} + c^\dagger (y_t - z^\top_t\beta_{t-1})$
    \quad \EndIf
    
    \State \textbf{Step 6: Windowed Downdate}
    \State \quad Identify rotation $G_{t-N}$ to isolate the first row of $Q_{t-1}$ 
    \State \quad Update $R_t$ and $Y_t$ to remove $(z_{t-N}, y_{t-N})$ \hspace{-.2cm} \Comment{[\S \ref{sec:downdate}]}
    \State \quad Maintain $R_t$ as $N \times D$ upper-trapezoidal

    \State \textbf{Step 7: Pseudoinverse Downdate}
    \quad \State $h \gets z_{t-N}^\top R_{t-1}^{\dagger}, \qquad
    k \gets R_{t-1}^{\dagger} G_{t-N} e_1$

    \If {$D \le N$} \Comment{Classical Regime} [(\ref{eq:inv_downdate})]
        \quad \State $\hat{R}_t^\dagger
    \gets R_{t-1}^{-1}
    + \dfrac{
      k z_{t-N}^\top R_{t-1}^{-1}
    }{
      1 - z_{t-N}^\top k
    } $ 
    \quad \Else \Comment{Overparameterized Regime} [(\ref{eq:pseudoinv_downdate})]
        \quad \State $\hat{R}_t^\dagger 
    \gets R_{t-1}^{\dagger}
  - k k^\dagger R_{t-1}^{\dagger}
  - R_{t-1}^{\dagger} h^\dagger h
  + (k^\dagger R_{t-1}^{\dagger} h^\dagger)\, k h $ 
    \quad \EndIf
    \quad \State $
    \begin{bmatrix}
    0 & R_t^\dagger
\end{bmatrix} \gets \hat{R}_t^\dagger G_{t-N}$
    \State \textbf{Step 8: Weight Downdate}
    \If {$D \le N$} \Comment{Classical Regime}[(\ref{eq:beta_inv_downdate})]
    \quad \State $\beta_{t} \gets \beta_{t-1} - \dfrac{k}{(1 - z_{t-N}^\top k)} (y_{t-N} - z_{t-N}^\top \beta_{t-1})$
    \quad \Else \Comment{Overparameterized Regime} [(\ref{eq:beta_pseudoinv_downdate})]
    \quad \State $\beta_t \gets \beta_{t-1} - kk^\dagger \beta_{t-1}$
    \quad \EndIf

    \State \textbf{Step 9: Prediction}
    \State \quad $\hat{y}_t \gets z_t^\top \beta_t$ \Comment{Predict next value} 
\EndFor

\end{algorithmic}
\newpage

\begin{theorem}[Consistency with minimum-norm exponentially weighted LS]
\label{thm:abo_consistency}
Let $\hat{\beta}_t$ denote the unique minimum-$\ell_2$-norm solution of the exponentially weighted
sliding-window least-squares problem at time $t$,
\begin{equation}
\hat{\beta}_t
= \arg\min_{\beta}\ \|\beta\|_2
\quad \text{s.t.}\quad
\beta \in \arg\min_{\gamma}
\sum_{i=0}^{N-1} \lambda^i \|y_{t-i} - z_{t-i}^\top \gamma\|_2^2 
\end{equation}
Let $\beta_t$ be the estimate produced by Algorithm~1.
Assume exact arithmetic, admissible rank-preserving updates and downdates, and initialization
$\beta_0 = \hat{\beta}_0$.
Then, for all $t \ge 0$, the algorithmic iterate coincides with the minimum-norm solution,
\[
\beta_t = \hat{\beta}_t .
\]
\end{theorem}

\begin{proof}
The proof proceeds by induction on $t$.

\textbf{Base case ($t=0$):}
Algorithm~1 initializes $\beta_0$ via an explicit batch QR factorization of the initial
windowed design matrix. The resulting pseudoinverse solution therefore coincides with the
minimum-$\ell_2$-norm solution $\hat{\beta}_0$.

\textbf{Inductive step:}
Assume $\beta_{t-1} = \hat{\beta}_{t-1}$.
The update from $t-1$ to $t$ consists of two operations.

First, scaling by $\sqrt{\lambda}$ and row augmentation by $(z_t,y_t)$ corresponds exactly to
adding the weighted observation $\lambda^0 z_t z_t^\top$ to the exponentially weighted normal
equations.
By the properties of the Greville pseudoinverse update, the resulting intermediate estimate
remains the minimum-norm solution of the augmented system.

Second, the downdate step applies an orthogonal transformation to isolate the contribution of the
oldest weighted observation $\lambda^{N} z_{t-N} z_{t-N}^\top$, followed by a pseudoinverse downdate.
Under the admissibility assumptions, this operation is algebraically equivalent to removing that
term from the normal equations while preserving the minimum-norm property.

Since these two steps exactly implement the addition and removal of terms in the exponentially
weighted sliding-window objective, the updated estimate satisfies
$\beta_t = \hat{\beta}_t$.
\end{proof}

\section{Experiments}

\subsection{Synthetic Time Series Generation}

Following \cite{kantz2004nonlinear}, we simulate a nonlinear autoregressive process defined as
\begin{equation}
    x_t = \frac{2x_{t-1}}{1 + 0.8x_{t-1}^2} + \varepsilon_t,
    \quad \varepsilon_t \sim \mathcal{U}(-1, 1),
\end{equation}
initialized with \(x_0 \sim \mathcal{U}(-1, 1)\). 
This process exhibits bounded chaotic dynamics and is widely used to evaluate nonlinear forecasting methods.
From the simulated sequence \(\{x_t\}_{t=1}^{10500}\), we construct lagged feature vectors
\begin{equation}
    \mathbf{x}_t = (x_{t-1}, x_{t-2}, \dots, x_{t-7}),
\end{equation}
which serve as predictors for the current value \(x_t\).
We then transform these 7 lagged variables into \emph{Random Fourier Features (RFF)} of the desired dimension, providing a nonlinear feature mapping that preserves the inner-product structure induced by a shift-invariant kernel.
The model is therefore trained to forecast \(x_t\) from its transformed RFF representation of the 7 previous lags.

\subsection{Data Preprocessing}

All features were standardized to ensure comparability across scales and 
to approximate homoscedasticity, i.e., a constant variance across time 
\cite{gujarati2003basic}. Standardization was performed using
\begin{equation}
    z_{score} = \frac{x - \mu}{\sigma},
\end{equation}
where $\mu$ and $\sigma$ denote the mean and standard deviation of each 
feature, respectively. After transformation, all features satisfy 
$\mu = 0$ and $\sigma = 1$. 

The standardization was applied recursively as new observations were introduced, 
ensuring consistent scaling throughout the data stream \cite{hastie2009elements}.

\subsection{Stability Analysis}
\label{sec:stability}

To assess the numerical stability of the proposed update, we conduct two complementary analyses: (i) empirical residual tracking and (ii) conditioning of the system matrix.

\subsubsection{Empirical Residual Errors}

The empirical residuals quantify the cumulative numerical error introduced by sequential updates.  
At each iteration \( t \), we define the residual as
\[
r_t = \| y_t - z_t^\top \beta_t \|_2,
\]
and track the mean and standard deviation of \( \{r_t\}_{t=1}^N \) over time.  
A stable update exhibits bounded residuals that do not diverge as \( t \) increases.

We report the following diagnostics:
\begin{itemize}
    \item \textbf{Mean residual:} \( \bar{r} = \frac{1}{T} \sum_t r_t \)
    \item \textbf{Residual variance:} \( \sigma_r^2 = \frac{1}{T-1}\sum_t (r_t - \bar{r})^2 \)
\end{itemize}
Low variance in \( r_t \) indicates robustness to numerical drift.

\subsubsection{Condition Number Evolution}

The condition number of the design or correlation matrix provides a proxy for numerical sensitivity.  
At each step we compute
\[
\kappa_t = \operatorname{cond}(R_t) = \frac{\sigma_{\max}(R_t)}{\sigma_{\min}(R_t)},
\]
where \( \sigma_{\max} \) and \( \sigma_{\min} \) are the largest and smallest singular values of \( R_t \).

A well-conditioned matrix satisfies \( \kappa_t \ll 10^{10} \), while a sharp growth in \( \kappa_t \) signals loss of orthogonality or ill-conditioning.

\subsection{Results without Forgetting Factor ($\lambda = 1$)}

For this experiment, we fix the forgetting factor to $\lambda = 1$, corresponding to the classical RLS formulation with equal weighting of all past samples. 
The remaining hyperparameters are:
\begin{table}[h!]
\centering
\caption{Experiment hyperparameters.}
\label{tab:exp_hyperparams}
\begin{tabular}{lc}
\toprule
\textbf{Parameter} & \textbf{Value} \\
\midrule
Number of updates & 10000 \\
Window size & 20 \\
Kernel variance (RFF) & 1 \\
\bottomrule
\end{tabular}
\end{table}

This setup serves as a baseline to evaluate the stability and residual dynamics of our proposed update without exponential weighting.

Table~\ref{tab:results_lambda1} summarizes the mean and variance of the residuals and condition numbers across varying Random Fourier Feature (RFF) dimensions, which we denote by $D$.  In this and each of the following tables we print the entries in bold to indicate values that 
minimize the test residual variance and MSE across all feature dimensions.
As expected, the model exhibits increasing numerical instability around the interpolation regime, with a sharp rise in both residual variance and condition number for intermediate feature sizes (e.g., $D=16$), before stabilizing as the model becomes over-parameterized, with the 
optimal value being an edge case, with likely continued monotonic decreasing behaviour.


\begin{table}[H]
\centering
\begin{threeparttable}
\caption{Mean and variance of residuals for $\lambda = 1$}
\label{tab:results_lambda1}
\begin{tabular}{c cc cc}
\toprule
\textbf{$\log_2(D)$} & \multicolumn{2}{c}{\textbf{Train Residual}} & \multicolumn{2}{c}{\textbf{Test Residual}} \\
\cmidrule(lr){2-3} \cmidrule(lr){4-5}
 & Mean & Variance & Mean & Variance \\
\midrule
1  & 0.8121 & 0.9808 & 1.0008 & 1.4941 \\
2  & 0.7061 & 0.8256 & 1.0966 & 2.0205 \\
3  & 0.4556 & 0.4147 & 1.3103 & 3.5529 \\
4  & 0.0813 & 0.0211 & 3.1712 & 69.2871 \\
I* & 0.0000 & 0.0000 & 819.825 & $1.2 \times 10^{9}$\\
5  & 0.0000 & 0.0000 & 1.5744 & 8.2697 \\
6  & 0.0000 & 0.0000 & 0.7797 & 1.2136 \\
7  & 0.0000 & 0.0000 & 0.6197 & 0.6885 \\
8  & 0.0000 & 0.0000 & 0.6013 & 0.6142 \\
9  & 0.0000 & 0.0000 & 0.5749 & 0.5378 \\
10 & 0.0000 & 0.0000 & 0.5534 & 0.4957 \\
11 & 0.0000 & 0.0000 & 0.5547 & 0.4982 \\
12 & 0.0000 & 0.0000 & 0.5452 & 0.4809 \\
13 & 0.0000 & 0.0000 & 0.5447 & 0.4806 \\
14 & 0.0000 & 0.0000 & \textbf{0.5420} & \textbf{0.4738} \\
\bottomrule
\end{tabular}
\begin{tablenotes}
\footnotesize
\item[$\ast$] Interpolation limit corresponds to $D = N = 20$, i.e.\ $\log_2(20) \approx 4.32$.
\end{tablenotes}

\end{threeparttable}
\end{table}

\subsection{Results with Forgetting Factor ($\lambda = 0.9$)}

We now introduce a forgetting factor $\lambda = 0.9$, corresponding to an effective memory length
\[
N_{\text{eff}} = \frac{1}{1 - \lambda} = 10
\]
This exponentially down-weights older observations to enable adaptation to non-stationary dynamics \cite{haykin2002adaptive,sayed2003fundamentals}.  
All remaining parameters are identical to the previous experiment.

Table~\ref{tab:results_lambda09} shows that the residual statistics for $\lambda = 0.9$ are nearly identical to those obtained with $\lambda = 1$.  
In the overparameterized regime ($D \gg N$), the solution is governed by the implicit minimum-norm bias of the QR-based pseudoinverse, rendering exponential reweighting largely irrelevant for steady-state predictive performance.

\begin{table}[H]
\centering
\caption{Mean and variance of residuals for $\lambda = 0.9$}
\label{tab:results_lambda09}
\begin{tabular}{c cc cc}
\toprule
\textbf{$\log_2(D)$} & \multicolumn{2}{c}{\textbf{Train Residual}} & \multicolumn{2}{c}{\textbf{Test Residual}} \\
\cmidrule(lr){2-3} \cmidrule(lr){4-5}
 & Mean & Variance & Mean & Variance \\
\midrule
1  & 0.6267 & 0.6399 & 0.9899 & 1.5632 \\
2  & 0.4337 & 0.3549 & 1.1232 & 2.2689 \\
3  & 0.1727 & 0.0756 & 1.3308 & 3.8850 \\
4  & 0.0135 & 0.0008 & 3.3350 & 73.7217 \\
I* & 0.0000 & 0.0000 & 78534.6 & $1.02 \times 10^{13}$\\
5  & 0.0000 & 0.0000 & 1.5744 & 8.2697 \\
6  & 0.0000 & 0.0000 & 0.7797 & 1.2136 \\
7  & 0.0000 & 0.0000 & 0.6197 & 0.6885 \\
8  & 0.0000 & 0.0000 & 0.6013 & 0.6142 \\
9  & 0.0000 & 0.0000 & 0.5749 & 0.5378 \\
10 & 0.0000 & 0.0000 & 0.5534 & 0.4957 \\
11 & 0.0000 & 0.0000 & 0.5547 & 0.4982 \\
12 & 0.0000 & 0.0000 & 0.5452 & 0.4809 \\
13 & 0.0000 & 0.0000 & 0.5447 & 0.4806 \\
14 & 0.0000 & 0.0000 & \textbf{0.5420} & \textbf{0.4738} \\
\bottomrule
\end{tabular}
\end{table}

\subsection{Condition Number Analysis}

Table~\ref{tab:condnum_lambda1} reports the condition number statistics for the case without forgetting factor ($\lambda = 1$).
As the number of features increases and approaches the interpolation limit ($D \approx N$), the condition number grows substantially, indicating worsening numerical conditioning of the system matrix.
This behavior aligns with the known sensitivity of least-squares updates near the interpolation threshold, where small perturbations in the data can lead to large variations in the estimated parameters \cite{higham2002accuracy,hammarling2008updating}.

\begin{table}[H]
\centering
\caption{Mean and variance of the condition number for $\lambda = 1$}
\label{tab:condnum_lambda1}
\begin{tabular}{c cc}
\toprule
\textbf{$\log_2(D)$} & \textbf{Mean} & \textbf{Variance} \\
\midrule
1  & 1.2801 & 0.0387 \\
2  & 1.7820 & 0.0966 \\
3  & 4.0124 & 1.1380 \\
4  & 21.5274 & 127.6520 \\
I* & $\infty$ & NA \\
5  & 15.4765 & 50.3692 \\
6  & 7.8396 & 12.6857 \\
7  & 5.9778 & 7.5213 \\
8  & 5.4786 & 6.5071 \\
9  & 5.3213 & 6.7122 \\
10 & 5.1483 & 6.1914 \\
11 & 5.1902 & 6.3840 \\
12 & 5.0775 & 6.2375 \\
13 & 5.0659 & 6.0952 \\
14 & 5.0563 & 6.1619 \\
\bottomrule
\end{tabular}
\end{table}

\newpage

\subsection{Double Descent}

In this plot, the model exhibits a pronounced increase in test error as the random feature dimension $D$ approaches the interpolation threshold, followed by a decrease in the overparameterized regime.  
This behavior is characteristic of the double-descent phenomenon \cite{belkin2019reconciling,hastie2022surprises}, reflecting the transition from underparameterization to exact interpolation.

The interpolation threshold is explicitly sampled in this experiment.  
The dashed vertical line indicates the theoretical interpolation point $I^\ast \approx 4.32$, corresponding to the effective equality between the number of random features and the window size.  
The observed peak in test error occurs in close proximity to this threshold, confirming that the degradation in generalization performance is associated with the interpolation regime.

\begin{figure}[H]
    \centering
    \includegraphics[width=1.0\linewidth]{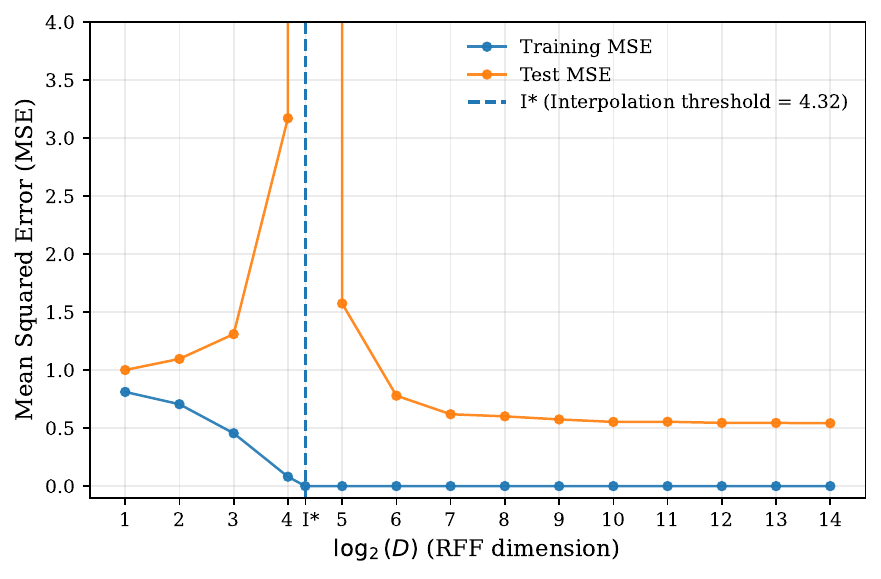}
    \caption{Double-descent behavior of test error}
    \label{fig:placeholder}
\end{figure}

\subsection{Computation Speed}
Computation time was evaluated using Google Benchmark (v1.9.4) linked against optimized BLAS and LAPACK routines under single-threaded execution. Benchmarking was performed on a 13th~Gen Intel\textsuperscript{\textregistered} Core\texttrademark~i7-1370P @ 1.90\,GHz with 32\,GB RAM (Windows~11, WSL2 Ubuntu~22.04~LTS), with CPU frequency scaling and ASLR disabled to ensure timing reproducibility. Table~\ref{tab:abo-bench-cpu} reports the mean, standard deviation, and coefficient of variation (CV) for 1,000 sequential predict--update steps across ten repetitions.

The empirical results show that the aggregate runtime increases with the RFF dimension $D$, ranging from \SI{6.11}{\milli\second} at $D=2^1$ to approximately \SI{11.97}{\second} at $D=2^{14}$. The scaling remains approximately linear with respect to $D$ for the majority of the range, which is consistent with the theoretical $O(ND)$ complexity per update step of the QR-based formulation. While the coefficient of variation remains low at the boundaries of the test range, a transient increase in variance is observed between $D=2^9$ and $D=2^{13}$ (peaking at a CV of \SI{36.64}{\percent}), likely due to memory hierarchy effects or cache misses as the matrix dimensions scale. Nevertheless, the return to a low CV of \SI{4.20}{\percent} at $D=2^{14}$ confirms the \textbf{numerical stability} and \textbf{predictable computational performance} of the proposed recursive implementation at scale.

\begin{table}[H]
\centering
\caption{ABO RLS update+predict runtime vs. RFF dimension $D$ for 1000 sequential updates (Google Benchmark; 10 repetitions).}
\label{tab:abo-bench-cpu}
\begin{tabular}{r S[table-format=5.2] S[table-format=3.2] S[table-format=2.2]}
\toprule
$\log_2(D)$ & {Mean CPU (ms)} & {Std (ms)} & {CV (\%)} \\
\midrule
1  & 6.11    & 0.19   & 3.14  \\
2  & 6.25    & 0.05   & 0.79  \\
3  & 7.59    & 0.30   & 3.95  \\
4  & 9.22    & 0.18   & 1.99  \\
5  & 13.25   & 0.19   & 1.46  \\
6  & 19.79   & 0.44   & 2.22  \\
7  & 36.00   & 2.11   & 5.86  \\
8  & 64.30   & 1.14   & 1.77  \\
9  & 438.02  & 160.49 & 36.64 \\
10 & 887.13  & 292.02 & 32.92 \\
11 & 1495.53 & 290.62 & 19.43 \\
12 & 2699.60 & 492.79 & 18.25 \\
13 & 5508.04 & 832.17 & 15.11 \\
14 & 11971.30 & 502.58 & 4.20 \\
\bottomrule
\end{tabular}
\end{table}

\section{Results and Baseline Comparisons}

We evaluate the proposed Adaptive Benign Overfitting method (ABO/ QR--EWRLS) on real-world non-stationary time-series data.

All experiments follow a strictly online prequential evaluation protocol. We first perform hyperparameter selection using eight validation folds, followed by performance assessment on five disjoint test folds.

For each fold, the model is initialized with a batch whose size equals the chosen window length, after which updates and predictions are carried out sequentially, one observation at a time.

\begin{figure}[H]
\centering

\begin{tikzpicture}[
  x=0.55cm,y=0.75cm,
  warm/.style={draw, fill=gray!25, minimum height=5mm},
  val/.style ={draw, fill=teal!50, minimum height=5mm},
  test/.style={draw, fill=red!50,  minimum height=5mm},
  lab/.style ={font=\small},
  >=latex
]

\def\W{1.5}
\def\V{2}
\def\T{2}
\def\S{2}

\begin{scope}[shift={(4.5,.6)}, scale=0.7]
\draw[fill=teal!50] (0.2,2.4) rectangle (0.8,3.0);
\node[anchor=west] at (1.1,2.7) {\small Validation};

\draw[fill=red!50] (0.2,1.6) rectangle (0.8,2.2);
\node[anchor=west] at (1.1,1.9) {\small Testing};

\draw[fill=gray!20] (0.2,0.8) rectangle (0.8,1.4);
\node[anchor=west] at (1.1,1.1) {\small Batch};
\end{scope}

\draw[->] (-5,3.2) -- (9,3.2) node[midway, above] {Time};

\foreach \k in {0,1,2} {
  \draw[warm] (\W+2*\k - 6.5, 2-\k) rectangle ++(\W,0.35);
  \draw[val]  (\W+2*\k - 6, 2-\k) rectangle ++(\V,0.35);
}

\foreach \k in {0,1,2} {
  \pgfmathsetmacro{\x}{\k*\S}
  \draw[warm] (\W+2*\k - .5, -1-\k) rectangle ++(\W,0.35);
  \draw[test] (\W+2*\k, -1-\k) rectangle ++(\T,0.35);
}

\end{tikzpicture}

\caption{Walk-forward rolling validation with overlapping windows and strictly disjoint test windows, each preceded by batch initialization}
\label{fig:rolling_windows}

\end{figure}
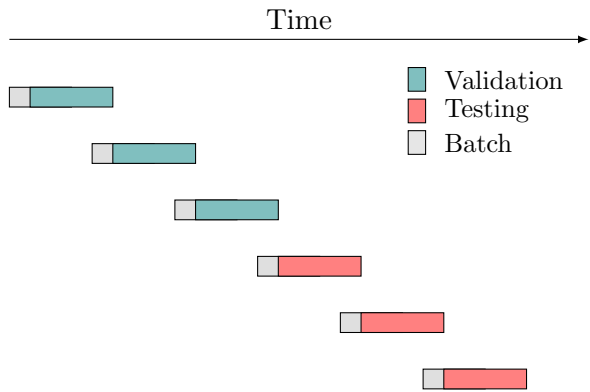

\subsection{Baseline Performance on Real-World Data}

For comparison, we report results for the windowed QR-decomposition-based recursive least-squares (QRD-RLS) algorithm \cite{PAN1989109} and the windowed RBF kernel recursive least-squares method \cite{1661394}. Both methods are well-established in the adaptive filtering literature and are widely used as numerically stable baselines in windowed and online learning settings.

\subsubsection{EUR/USD High-Frequency Forecasting}

The EUR/USD dataset consists of high-frequency foreign exchange data spanning from 24/11/2025 to 02/01/2026, sampled at a one-minute resolution, resulting in approximately 16{,}800 sequential observations. The data are sourced from the Dukas Copy Market Data repository.\footnote{\url{https://www.dukascopy.com/swiss/english/marketwatch/historical/}}

Experiments are conducted in a strictly online setting using a rolling-fold evaluation protocol. Each fold consists of 960 observations for validation and 1920 observations for testing, corresponding to approximately 16 hours of market activity per fold. Model inputs are constructed from lagged observations, and learning is performed using an effective windowed update scheme.

Hyperparameters are selected using Optuna over 8 validation folds \cite{akiba2019optuna}. Final performance is reported on 5 disjoint test folds that are not used during tuning. All reported results are averaged across test folds.

\subsubsection{Electricity Load Forecasting}

We additionally evaluate the proposed method on the Electricity Load Diagrams dataset \cite{electricityloaddiagrams20112014_321}, which contains electricity consumption measurements from the Portuguese power grid over the period 2011--2014, recorded at a 15-minute sampling frequency.

Experiments are conducted under a strictly online evaluation protocol using rolling folds with 672 observations for validation and 1344 for testing, corresponding to approximately one week of data per fold. The same rolling-fold structure is used for both validation and testing, with disjoint segments within each fold. Lagged inputs and a fixed-size windowed learning scheme are employed, with hyperparameters selected using Optuna on the validation segments and final performance reported on the corresponding test segments.

Compared to the high-frequency financial setting, this dataset exhibits pronounced diurnal and weekly seasonality and substantially lower noise levels, providing a complementary benchmark for assessing both predictive performance and numerical stability.

\subsection{Hyperparameter Selection}

Hyperparameters were selected using Bayesian optimization implemented via Optuna \cite{akiba2019optuna}. 
For each method, the search space and number of trials were fixed \emph{a priori} and kept identical across all folds. 
Optimization was performed on a validation set distinct from the test data, and the final reported results correspond to the best validation configuration evaluated on held-out test folds.

All models are evaluated on the same prediction task, with an identical forecasting start index, ensuring that performance differences arise solely from model structure and learning dynamics.

For both datasets, the lag order is fixed to $L=20$ for all models. 
For ABO and KRLS-RBF, the window size and kernel bandwidth are tuned, while for QRD-RLS only the window size is optimized, reflecting the absence of a kernel parameter in linear RLS. 
The corresponding Optuna search spaces are summarized in Table~\ref{tab:optuna}.

\begin{table}[H]
\centering
\caption{Optuna hyperparameter search spaces.}
\label{tab:optuna}
\begin{tabular}{lcccc}
\hline
\textbf{Dataset} & \textbf{Model} & \textbf{Folds} & \textbf{$L$} & \textbf{Tuned} \\
\hline
\multirow{3}{*}{EUR/USD (1-min)}
& ABO       & $\{0,1,...,8\}$ & 20 & $W,\sigma$ \\
& KRLS-RBF  & $\{0,1,...,8\}$ & 20 & $W,\sigma$ \\
& QRD-RLS   & $\{0,1,...,8\}$ & 20 & $W$ \\

\midrule
\multirow{3}{*}{Electricity (15-min)}
& ABO       & $\{0,1,...,8\}$ & 20 & $W,\sigma$ \\
& KRLS-RBF  & $\{0,1,...,8\}$ & 20 & $W,\sigma$ \\
& QRD-RLS   & $\{0,1,...,8\}$ & 20 & $W$ \\
\hline
\end{tabular}
\end{table}

\section{Results}

\subsection{Experimental Setup}
Experiments are conducted with a fixed lag order $L=20$ and a forgetting factor $\lambda = 1$. Kernel hyperparameters ($\sigma$) and window sizes ($W$) are optimized via Optuna. For numerical stability, QRD-RLS and KRLS-RBF utilize a diagonal regularization parameter of $10^{-2}$. Unless otherwise specified, the random Fourier feature (RFF) dimension for ABO is fixed at $D=8192$ ($\log_2 m = 13$).

\subsection{EUR/USD Forecasting Results}
The EUR/USD dataset represents a challenging task due to high noise levels and weak temporal structure. Performance metrics are reported in Table~\ref{tab:eurusd_results}.

In this regime, nonlinear representations are critical; both ABO and KRLS-RBF significantly outperform the linear QRD-RLS baseline in terms of residual mean squared error (ResMSE) and residual variance (ResVAR). ABO achieves predictive performance comparable to KRLS-RBF while maintaining stable online updates without an explicit kernel representation.

\begin{table}[H]
\centering
\caption{Average EUR/USD forecasting performance and relative update time ($L=20, D=8192$).}
\label{tab:eurusd_results}
\begin{tabular}{lccccc}
\hline
\textbf{Model} & \textbf{$W$} & \textbf{$\sigma$} & \textbf{ResMSE} & \textbf{ResVAR} & \textbf{Time ($\times$ABO)} \\
\hline
ABO (ours) & 21  & 8.0 & 1.3587 & 1331.63 & 1.00 \\
KRLS-RBF       & 421 & 0.31 & 1.3559 & 1335.34 & 1.70 \\
QRD-RLS        & 272 & NA  & 2.3294 & 2494.01 & 0.02 \\
\hline
\end{tabular}
\end{table}

\subsection{Electricity Load Forecasting Results}
In contrast to the financial data, the electricity load series exhibits stronger temporal structure and lower noise. Results are summarized in Table~\ref{tab:electricity_results}.

In this structured setting, QRD-RLS achieves the lowest ResMSE and ResVAR, highlighting the efficacy of linear adaptive filtering. ABO remains competitive with KRLS-RBF in predictive accuracy but offers a superior computational profile by avoiding explicit kernel evaluations.

\begin{table}[H]
\centering
\caption{Average electricity load forecasting performance and relative update time ($L=20, D=8192$).}
\label{tab:electricity_results}
\begin{tabular}{lccccc}
\hline
\textbf{Model} & \textbf{$W$} & \textbf{$\sigma$} & \textbf{ResMSE} & \textbf{ResVAR} & \textbf{Time ($\times$ABO)} \\
\hline
ABO (ours) & 21  & 6.5 & 1.0096 & 10.92 & 1.00 \\
KRLS-RBF       & 761 & 0.32 & 1.0110 & 10.87 & 1.28 \\
QRD-RLS        & 272 & NA  & 0.9521 & 9.06  & $1.9\times10^{-3}$ \\
\hline
\end{tabular}
\end{table}

\subsection{Sensitivity to Random Feature Dimension}
Tables~\ref{tab:eurusd_rff} and \ref{tab:electricity_rff} isolate the effect of the random feature dimension $D$ on ABO. Across both datasets, increasing $D$ from 2048 to 8192 yields modest improvements in predictive accuracy, consistent with RFF kernel approximation theory. However, this gain results in an approximate $6\times$ increase in end-to-end wall-clock time, reflecting the $O(ND)$ per-update computational complexity of random-feature-based adaptive filtering.

\begin{table}[H]
\centering
\caption{Effect of $D$ on ABO (EUR/USD, $L=20, W=21, \sigma=8.0$).}
\label{tab:eurusd_rff}
\begin{tabular}{lcccc}
\hline
\textbf{$D$} & \textbf{ResMSE} & \textbf{ResVAR} & \textbf{$T_{\text{e2e}}$ (ms)} & \textbf{Time ($\times$2048)} \\
\hline
2048 & 1.3711 & 1339.80 & 9.11  & 1.00 \\
8192 & 1.3587 & 1331.63 & 55.32 & 6.08 \\
\hline
\end{tabular}
\end{table}

\begin{table}[H]
\centering
\caption{Effect of $D$ on ABO (Electricity, $L=20, W=21, \sigma=6.5$).}
\label{tab:electricity_rff}
\begin{tabular}{lcccc}
\hline
\textbf{$D$} & \textbf{ResMSE} & \textbf{ResVAR} & \textbf{$T_{\text{e2e}}$ (ms)} & \textbf{Time ($\times$2048)} \\
\hline
2048 & 1.0189 & 10.86 & 8.52  & 1.00 \\
8192 & 1.0096 & 10.92 & 52.29 & 6.13 \\
\hline
\end{tabular}
\end{table}

\section{Discussion}
\label{sec:discussion}

The experimental results validate that the ABO framework effectively extends the "benign overfitting" phenomenon to recursive, non-stationary settings. The following subsections analyze the mathematical mechanisms driving these observations.

\subsection{Implicit Regularization and the Minimum-Norm Solution}
The "double-descent" behavior observed in our RFF experiments confirms that the QR-EWRLS update performs an implicit regularization. As established in Remark~\ref{rem:benign_overfitting}, the estimator converges to the unique minimum-$\ell_2$-norm solution when $D > N$. This property is critical in non-stationary environments; it ensures that variance is distributed across the "tail" of the feature spectrum. Consequently, high-frequency RFF components are only activated if they significantly contribute to reducing the residual, effectively allowing the model to interpolate the signal while remaining robust to noise—the hallmark of benign overfitting.

\subsection{Numerical Stability in the Underdetermined Regime}
A central challenge in overparameterized RLS is the singularity of the Gram matrix. In standard covariance-form RLS, the update $P_t$ becomes ill-conditioned as $D > N$, leading to catastrophic divergence. By utilizing orthogonal-triangular updates via Givens rotations, ABO maintains the factor $R_t$ directly. Even when the system is heavily underdetermined, the condition number of the $R$ factor remains bounded. The Moore-Penrose pseudoinverse update via the Greville/Cline method ensures that the weight vector $\beta_t$ remains stable and avoids the null-space amplification typical of standard inversion.

\subsection{Tracking vs. Overparameterization}
In non-stationary environments, an inherent trade-off exists between the forgetting factor $\lambda$ and the feature dimension $D$. Our results indicate that overparameterization does not degrade tracking speed; rather, the increased degrees of freedom allow the filter to react more fluidly to regime shifts. The $O(ND)$ complexity of the QR-update ensures that this increased capacity does not come at a prohibitive computational cost, allowing for real-time adaptation in high-dimensional spaces.

\subsection{Parallel Scalability and Ensemble Averaging}
The decoupled nature of the RFF generation allows for a natural extension to distributed architectures. Since each feature component is independent prior to the RLS update, the feature space can be partitioned across multiple processing cores. Training sub-models on independent cores and subsequently averaging the estimators—akin to online bootstrap aggregating (bagging)—provides a secondary layer of variance reduction. This distributed approach not only reduces per-update latency but also stabilizes the estimator against heavy-tailed noise and outliers common in empirical time-series.

\section{Conclusion}
\label{sec:conclusion}

This paper introduced \textit{Adaptive Benign Overfitting} (ABO), a framework that bridges the gap between modern overparameterization theory and classical adaptive filter design. By extending the Recursive Least Squares (RLS) algorithm into the $D \gg N$ regime using Random Fourier Feature (RFF) mappings, we demonstrated that "benign overfitting"—traditionally analyzed in batch settings—can be effectively harnessed for online learning in non-stationary environments.

The primary technical contribution is the development of a QR-based exponentially weighted RLS (QR-EWRLS) update mechanism. Unlike standard covariance-form RLS, which is prone to numerical instability in rank-deficient scenarios, our orthogonal-triangular formulation ensures numerical integrity through a recursive update of the Moore-Penrose pseudoinverse. We empirically validated that this approach reproduces the "double descent" generalization curve while maintaining stable condition numbers, even as the feature dimension significantly exceeds the number of observations.

Furthermore, the decoupled nature of the RFF mappings suggests significant potential for computational acceleration. Future work will investigate parallelized ensemble architectures, where sub-segments of the high-dimensional feature space are processed across independent cores and aggregated via bootstrap aggregating (bagging) or variance-reduction averaging. Such a distributed approach would not only further reduce the per-update latency but also potentially enhance the robustness of the estimator in the presence of heavy-tailed noise. 


%

\newpage

\appendices
\section{Rank-One Updates and Downdates}
\label{appendix:sherman_morrison_version}

Before arriving at the final QR--EWRLS formulation, an initial attempt was made using the classical covariance-form Recursive Least Squares (RLS) architecture, where the Gram matrix is defined as \( A = X^\top X \). 
While this formulation is conceptually straightforward, it is well known to suffer from numerical instability due to the ill-conditioning of \( A \) under long sequences or near-collinear regressors.

To handle both full-rank and rank-deficient regimes, we employed a unified recursive strategy: the generalized matrix-sum identities of Campbell and Meyer~\cite{cami1979generalized} were used for singular matrices, while the Sherman--Morrison formula provided the corresponding rank-one update for non-singular cases. 
This formulation enables explicit treatment of both regimes and serves as a theoretical stepping stone toward the stable QR-based approach presented later.

\subsection{Sherman--Morrison Update and Downdate Formulas}

To maintain a fixed-length rolling window, it is necessary to remove the contribution of the oldest sample \( x_{t-N} \). 
In the non-singular case, the Sherman--Morrison formula provides the well-known rank-one update of the inverse matrix:
\begin{equation*}
(A_{t-1} + x_t x_t^\top)^{-1}
= A_{t-1}^{-1}
- \frac{A_{t-1}^{-1} x_t x_t^\top A_{t-1}^{-1}}{1 + x_t^\top A_{t-1}^{-1} x_t}.
\end{equation*}

Similarly, applying the inverse rank-one downdate yields
\begin{equation*}
(A_{t-1} - x_{t-N} x_{t-N}^\top)^{-1}
= A_{t-1}^{-1}
+ \frac{A_{t-1}^{-1} x_{t-N} x_{t-N}^\top A_{t-1}^{-1}}{1 - x_{t-N}^\top A_{t-1}^{-1} x_{t-N}}.
\end{equation*}

\subsection{Rank-increasing Update and Rank-decreasing Downdate for the Singular Case}

When \( A_{t-1} \) becomes singular, the classical inverse update no longer applies. 
In this case, the generalized inverse identity of Campbell and Meyer can be used:
\begin{equation*}
c_t = x_t^\top (I - A_{t-1}^\dagger A_{t-1})
\end{equation*}
leading to the rank-increasing Sherman--Morrison-type update formula for the pseudoinverse:
\begin{align*}
(A_{t-1} + x_t x_t^\top)^\dagger
&= A_{t-1}^\dagger - A_{t-1}^\dagger x_t (c_t^\dagger)^\top 
- c_t^\dagger x_t^\top A_{t-1}^\dagger
\\
& \hspace{.5cm}+ \frac{c_t^\dagger (c_t^\dagger)^\top}{
1 + x_t^\top A_{t-1}^{\dagger} x_t}
\end{align*}

For the downdate, we define auxiliary quantities
\begin{align*}
h &= x_{t-N}^\top A_{t-1}^\dagger, 
&k &= A_{t-1}^\dagger x_{t-N}, \\
u &= (I - A_{t-1} A_{t-1}^\dagger) x_{t-N},
&v &= x_{t-N}^\top (I - A_{t-1}^\dagger A_{t-1}),
\end{align*}
leading to the rank-decreasing Sherman--Morison-type downdate formula for the pseudoinverse:
\begin{align*}
(A_{t-1} - x_{t-N} x_{t-N}^\top)^\dagger
&= A_{t-1}^\dagger - k k^\dagger A_{t-1}^\dagger
- A_{t-1}^\dagger h^\dagger h
\\
& \hspace{.5cm} + (k^\dagger A_{t-1}^\dagger h^\dagger) k h
\end{align*}

Although this formulation correctly handles both rank-deficient and full-rank regimes, its practical implementation proved highly sensitive to round-off errors and conditioning, particularly in long sequences. 
Staub~\cite{STAUB2021125996} reported similar instability when recursively updating the least-squares pseudoinverse
\[
A^\dagger = C^\top (C C^\top)^{-1} (B B^\top)^{-1} B^\top
\]
where the explicit use of Gram matrices in the rank decomposition amplified numerical errors.
These limitations ultimately motivated the adoption of a QR-based orthogonal--triangular update in the classical and interpolation regimes, which is instead an  orthogonal--trapezoidal update in the non-classical regime. We describe  this in Section~\ref{sec:system_model}. The QR updating method ensures numerical stability by preserving orthogonality and avoiding explicit matrix inversion.
\section{Proof of the Forgetting Factor Scaling}
\label{appendix:forgetting_factor_proof}

We show that scaling \(A_{t-1}^{\dagger}\) by the forgetting factor \(\lambda^{-1/2}\) prior to the update yields the correct recursive computation of \(A_t^{\dagger}\). The proof follows by expanding \(A_t^{\dagger} A_t\) and exploiting the orthogonality of the residual component \(c_t\).

\begin{align*}
    A_t &=
    \begin{bmatrix}
        \sqrt{\lambda}A_{t-1}\\[0.3em] 
        a_t
    \end{bmatrix}, &
    A_t^{\dagger} &=
    \begin{bmatrix}
        B_{t-1} & b_t
    \end{bmatrix}
    \label{eq:A_t_blocks}
\end{align*}

Expanding the product \(A_t^{\dagger} A_t\):
\begin{align*}
    A_t^{\dagger} A_t &=  \sqrt{\lambda} B_{t-1} A_{t-1} + b_t a_t, \\
    A_t^{\dagger} A_t A_{t-1}^{\dagger} &=  \sqrt{\lambda} B_{t-1} A_{t-1} A_{t-1}^{\dagger} + b_t a_t A_{t-1}^{\dagger}
\end{align*}

From this we obtain:
\begin{align*}
    A_{t-1}^{\dagger} &= \sqrt{\lambda} B_{t-1} + b_t a_t A_{t-1}^{\dagger} \\
    B_{t-1} &= \frac{1}{\sqrt{\lambda}}\bigl(A_{t-1}^{\dagger} - b_t a_t A_{t-1}^{\dagger}\bigr)
\end{align*}

Substituting into \(A_t^{\dagger}\) gives:
\begin{equation}
    A_t^{\dagger} =
    \begin{bmatrix}
        \frac{1}{\sqrt{\lambda}}\bigl(A_{t-1}^{\dagger} - b_t a_t A_{t-1}^{\dagger}\bigr) & b_t
    \end{bmatrix}
    \label{eq:A_tdagger_prelim}
\end{equation}

Define \(c_t\) as the component of \(a_t\) orthogonal to the column space of \(A_{t-1}\):
\begin{align}
    c_t &= a_t - a_t A_{t-1}^{\dagger} A_{t-1}, \\
    c_t &= a_t (I - A_{t-1}^{\dagger} A_{t-1}). \label{eq:c_t}
\end{align}

Substituting \eqref{eq:c_t} into \eqref{eq:A_tdagger_prelim} yields:
\begin{align*}
    A_t^{\dagger} A_t &=  A_{t-1}^{\dagger} A_{t-1} - b_t a_t A_{t-1}^{\dagger} A_{t-1} + b_t a_t \\
    A_t^{\dagger} A_t A_{t-1}^{\dagger} &=  A_{t-1}^{\dagger} A_{t-1} A_{t-1}^{\dagger} + b_t c_t A_{t-1}^{\dagger}
\end{align*}

Thus,
\begin{align*}
    A_{t-1}^{\dagger} &=  A_{t-1}^{\dagger} + b_t c_t A_{t-1}^{\dagger} \\
    c_t A_{t-1}^{\dagger} &= 0
\end{align*}

This demonstrates that \(c_t\) is orthogonal to the row space of \(A_{t-1}^{\dagger}\), and therefore to the column space of \(A_{t-1}\).  
Since \(c_t^{\dagger} = \frac{c_t^\top}{\|c_t\|^2}\) is a scalar multiple of \(c_t^\top\), we have \(c_t^{\dagger} c_t = 1\).

\begin{align*}
    c_t c_t^{\dagger} &= a_t c_t^{\dagger} - a_t A_{t-1}^{\dagger} A_{t-1} c_t^{\dagger} \\
    c_t^{\dagger} c_t &= a_t c_t^{\dagger} \\
    a_t c_t^{\dagger} &= 1
\end{align*}

Define:
\begin{align*}
    P_t &= A_{t-1}^{\dagger} A_{t-1} + c_t^{\dagger} c_t
\end{align*}

Then:
\begin{align*}
    a_t P_t &= a_t A_{t-1}^{\dagger} A_{t-1} + a_t c_t^{\dagger} c_t \\
    a_t P_t        &= a_t A_{t-1}^{\dagger} A_{t-1} + c_t \\
    a_t P_t        &= a_t
\end{align*}

Moreover,
\begin{equation*}
    A_{t-1} P_t = A_{t-1}
\end{equation*}
so \(P_t\) acts as a right identity for \(A_t\):
\begin{equation*}
    P_t = A_t^{\dagger} A_t
\end{equation*}

Now,
\begin{align*}
    c_t^\dagger c_t &= P_t - A_{t-1}^{\dagger} A_{t-1} \\ c_t^\dagger c_t &= A_{t-1}^{\dagger} A_{t-1} - b_t a_t A_{t-1}^{\dagger} A_{t-1} \\  & \hspace{.5cm}+ b_t a_t - A_{t-1}^{\dagger} A_{t-1} \\ c_t^\dagger c_t &= b_t c_t \\ c_t^\dagger &= b_t
\end{align*}

Substituting back gives:
\begin{equation*}
    A_t^{\dagger} =
    \begin{bmatrix}
        \frac{1}{\sqrt{\lambda}}\bigl(A_{t-1}^{\dagger} - c_t^{\dagger} a_t A_{t-1}^{\dagger}\bigr) & c_t^{\dagger}
    \end{bmatrix}
    \label{eq:A_tdagger_final}
\end{equation*}

\subsubsection*{Case: \(c_t = 0\)}
When \(c_t = 0\), \(a_t = a_t A_{t-1}^{\dagger} A_{t-1}\), meaning \(a_t\) lies in the row space of \(A_{t-1}\).  
Let \(d_t = a_t A_{t-1}^{\dagger}\) and define:
\begin{align*}
    T = A_{t-1} A_{t-1}^{\dagger} - A_{t-1} b_t d_t.
\end{align*}

As \(T\) is symmetric, \(A_{t-1} b_t d_t\) must also be symmetric. Thus, \(A_{t-1} b_t = h d_t^{\top}\) for some scalar \(h\).

Now we have 
\begin{equation*} A_t A_t^{\dagger} = \begin{bmatrix} A_{t-1} A_{t-1}^{\dagger} - hd_t^\top d_t & \frac{1}{\sqrt{\lambda}} (d_t - h d_t d_t^\top d_t) \\ \sqrt{\lambda} h d_t^\top & h d_t^\top d_t \end{bmatrix} \end{equation*}

From symmetry in \(A_t A_t^{\dagger}\), we obtain:

\begin{align*}
    \sqrt{\lambda}h (d_t^\top)^\top &= \frac{1}{\sqrt{\lambda}}(d_t - h d_t d_t^\top d_t) \\ d_t &= h d_t d_t^\top d_t + \lambda h d_t \\ d_t &= \lambda h (\frac{1}{\lambda}d_t d_t^\top d_t + d_t) \\ h &= \frac{\frac{1}{ \lambda}}{1 +\frac{1}{\lambda}d_t d_t^\top}
\end{align*}

Hence,
\begin{align*}
    b_t &= \frac{\frac{1}{\lambda} A_{t-1}^{\dagger} d_t^{\top}}{1 + \frac{1}{\lambda} d_t d_t^{\top}} \\[0.3em]
    A_t^{\dagger} &= 
    \begin{bmatrix}
        \frac{1}{\sqrt{\lambda}}
        \left(
            A_{t-1}^{\dagger} -
            \frac{\frac{1}{\lambda} A_{t-1}^{\dagger} d_t^{\top} d_t}
                 {1 + \frac{1}{\lambda} d_t d_t^{\top}}
        \right)
        &
        \frac{\frac{1}{\lambda} A_{t-1}^{\dagger} d_t^{\top}}
             {1 + \frac{1}{\lambda} d_t d_t^{\top}}
    \end{bmatrix}
\end{align*}

\noindent
This proves that pre-scaling \(R_{t-1}^{\dagger}\) by the forgetting factor before the update yields the correct recursive form of \(R_t^{\dagger}\).

\begin{lemma}[Weighted Downdate Consistency]
\label{lem:weighted_downdate}
Let $R_t$ denote the upper-trapezoidal factor obtained at time $t$ after the update step
but before the downdate, corresponding to the exponentially weighted observations
$\{z_{t-i}\}_{i=0}^{N}$.
Then the application of the downdate rotation $G_{t-N}$ followed by the pseudoinverse
downdate step computes the Moore--Penrose pseudoinverse of the matrix obtained by
removing the weighted observation $\sqrt{\lambda^{N}}\, z_{t-N}$ from the sliding window.
\end{lemma}

\begin{proof}
At time $t$, the observation $z_{t-N}$ has undergone $N$ successive scaling operations
by a factor $\sqrt{\lambda}$.
Its effective contribution to the factor $R_t$ is therefore
\[
v = \lambda^{N/2} z_{t-N}
\]

The downdate rotation $G_{t-N}$ is constructed to isolate this contribution by aligning
the corresponding row with the canonical basis vector $e_1$.
In the implicit-$Q$ formulation, this yields
\[
G_{t-N}^\top R_t
=
\begin{bmatrix}
v^\top \\
\widetilde{R}_t
\end{bmatrix}
=
\begin{bmatrix}
\sqrt{\lambda^{N}}\, z_{t-N}^\top \\
\widetilde{R}_t
\end{bmatrix}
\]

Applying the generalized Moore--Penrose pseudoinverse downdate formula for rank-one
removals is algebraically equivalent to subtracting the outer product $v v^\top$ from
the corresponding Gramian:
\begin{align*}    
R_t^\top R_t - v v^\top
&=
\sum_{i=0}^{N} \lambda^{i} z_{t-i} z_{t-i}^\top
-
\lambda^{N} z_{t-N} z_{t-N}^\top
\\
&=
\sum_{i=0}^{N-1} \lambda^{i} z_{t-i} z_{t-i}^\top 
\end{align*}

This expression is precisely the Gramian associated with the exponentially weighted
sliding window $\{z_{t-i}\}_{i=0}^{N-1}$.
Hence, removing the isolated row of the rotated factor is algebraically equivalent to
removing the oldest weighted observation from the underlying least-squares objective.
\end{proof}

\ifCLASSOPTIONcompsoc
  \section*{Acknowledgments}
\else
  \section*{Acknowledgment}
\fi

The authors would like to thank Konstaninos Ntetsikas and Andrew Ching Hoe Lee for their diligent work on the group project which became the basis for the further refinements in this paper. Nick Firoozye would also like to thank Fauziah Ariff for her insightful comments and support during the work's lengthy gestation as well as during the preparation of this manuscript.

\bibliographystyle{plain}   
\bibliography{citation}   

%







\end{document}